\documentclass[english,twocolumn,transaction,10pt]{IEEEtran}
\usepackage[T1]{fontenc}
\usepackage{float}
\usepackage{amsmath}
\usepackage{amsthm}
\usepackage{amssymb}
\usepackage{stackrel}
\usepackage{graphicx}
\usepackage{setspace}
\usepackage{url}
\usepackage{bm}
\usepackage{caption}
\usepackage{subfigure}
\usepackage{xcolor}
\usepackage{cases}
\usepackage{algorithm}
\usepackage{algorithmic}
\usepackage{stfloats}

\makeatletter

\floatstyle{ruled}
\newfloat{algorithm}{tbp}{loa}
\providecommand{\algorithmname}{Algorithm}
\floatname{algorithm}{\protect\algorithmname}

\theoremstyle{plain}

\theoremstyle{definition}

\theoremstyle{plain}

\theoremstyle{definition}

\theoremstyle{plain}

\newtheorem{assumption}{Assumption}

\newtheorem{theo}{Theorem}

\newtheorem{proposition}{Proposition}
\IEEEoverridecommandlockouts
\usepackage{ifpdf}
\usepackage{cite}
\ifCLASSOPTIONcompsoc
\usepackage[caption=false,font=normalsize,labelfont=sf,textfont=sf]{subfig}
\else
\usepackage[caption=false,font=footnotesize]{subfig}
\fi
\usepackage{amsmath}
\usepackage{mathtools}
\usepackage{booktabs}
\usepackage{multirow}
\usepackage{setspace}
\usepackage{lettrine}
\usepackage{array}
\@ifundefined{showcaptionsetup}{}{
	\PassOptionsToPackage{caption=false}{subfig}}
\usepackage{subfig}
\makeatother
\allowdisplaybreaks[4]
\usepackage{diagbox}
\usepackage{babel}
\usepackage{pifont}

\newtheorem{lemma}{Lemma}

\begin{document}
\captionsetup[figure]{font={small}, name={Fig.}, labelsep=period}

\title{Route-and-Aggregate Decentralized Federated Learning Under Communication Errors
}
\author{Weicai~Li, \textit{Graduate Student Member, IEEE},~Tiejun~Lv, \textit{Senior Member, IEEE},~Wei~Ni, \textit{Fellow, IEEE},\\  
 Jingbo~Zhao,
 Ekram~Hossain, \textit{Fellow, IEEE}, and H. Vincent Poor, \textit{Life Fellow, IEEE}

\thanks{W. Li, T. Lv and J. Zhao are with the School of Information and Communication Engineering, Beijing University of Posts and Telecommunications (BUPT), Beijing 100876, China. W. Li was also with the School of Electrical and Data Engineering, University of Technology Sydney, NSW, Australia; (e-mail: \{liweicai, lvtiejun, zhjb\}@bupt.edu.cn). 

W.~Ni is with the School of Electrical and Data Engineering, University of Technology Sydney, NSW 2007, Australia, and the School of Computing Science and Engineering, University of New South Wales, NSW 2052, Australia. (e-mail: wei.ni@ieee.org). 

E. Hossain is with the Department of Electrical and Computer Engineering, University of Manitoba, Canada (e-mail: ekram.hossain@umanitoba.ca). 

H. V. Poor is with the Department of Electrical and Computer Engineering, Princeton University, Princeton, NJ 08544, USA (e-mail: poor@princeton.edu).}	
	}

	\maketitle
\begin{abstract}
Decentralized federated learning (D-FL) allows clients to aggregate learning models locally, offering flexibility and scalability. 
Existing D-FL methods use gossip protocols, which are inefficient when not all nodes in the network are D-FL clients. 
This paper puts forth a new D-FL strategy, termed Route-and-Aggregate (R\&A) D-FL, where participating clients exchange models with their peers through established routes (as opposed to flooding) and adaptively normalize their aggregation coefficients to compensate for communication errors. 
The impact of routing and imperfect links on the convergence of R\&A D-FL is analyzed, revealing that convergence is minimized when routes with the minimum end-to-end packet error rates are employed to deliver models. 
Our analysis is experimentally validated through three image classification tasks and two next-word prediction tasks, utilizing widely recognized datasets and models.
R\&A D-FL outperforms the flooding-based D-FL method in terms of training accuracy by 35\% in our tested 10-client network, and shows strong synergy between D-FL and networking. 
In another test with 10 D-FL clients, the training accuracy of R\&A D-FL with communication errors approaches that of the ideal C-FL without communication errors, 
as the number of routing nodes (i.e., nodes that do not participate in the training of D-FL) rises to~28.

\end{abstract}
	
\begin{keywords}
Decentralized federated learning, peer-to-peer network, imperfect channel, routing, convergence.
\end{keywords}

	\section{Introduction}
Being a new paradigm of federated learning (FL), decentralized FL (D-FL), allowing participants to train a shared model by exchanging updates directly, without relying on a central server, has been 
increasingly studied for scenarios where central aggregators are absent~\cite{1638541,10038786,10323597} or single-point failures need to be avoided~\cite{9687521}. 
D-FL has significant potential for industrial applications, building on the widespread acceptance and proven utility of centralized FL (C-FL). D-FL overcomes the primary limitation of C-FL, i.e., the reliance on a central node, while preserving its key advantages, such as parallel training, local data and model training, and enhanced data privacy. By eliminating the need for central coordination, D-FL offers greater scalability, fault tolerance, and robustness, when applied to the industrial applications where C-FL has been considered, e.g., healthcare for privacy-preserving diagnostics~\cite{9394758}, finance for fraud detection~\cite{WEST201647}, manufacturing for predictive maintenance~\cite{s23177331,SAEZDECAMARA2023103299,ZONTA2020106889}, and smart cities for optimizing energy and traffic systems~\cite{6348251}.

The network topology and link quality (e.g., communication errors) can have a significant impact on the performance of any D-FL system. 
For practical reasons, not all nodes have to participate in D-FL in a network. It is essential to limit participation to a specific subset of computationally capable and engaged clients, rather than involving all nodes in the network. 
To mitigate security concerns~\cite{hui2022horizontal}, typically, a client delivers its model to peers (i.e., other participating clients) via a secure communication channel, such as Transport Layer Security (TLS)\cite{hui2022horizontal}, rather than broadcasting it indiscriminately~\cite{9716792}. Therefore, implementing an effective routing strategy to distribute local models among clients, while ensuring the consistency requirements of D-FL adapting to network topology and link quality, is critical but has not been studied in the existing literature.

Existing D-FL methods have primarily employed gossip-based protocols under the assumption of static network topologies with all nodes participating in D-FL in a network. 
{\color{black}
The primary focus of the existing studies has been on a gossip protocol~\cite{9716792,9563232}, referred to as ``Aggregate-as-You-Go (AaYG)'' in this paper. 
The existing studies allow clients and their one-hop~\cite{9563232,8950073} neighbors to exchange their models and aggregate the models locally by designing consensus protocols. Less connected clients would incur inconsistent local models (substantially differing from the rest of the clients).}
While gossip-based protocols, such as epidemic (flooding-like) gossip, enable decentralized operation, they suffer from slow convergence and limited ability to provide targeted delivery to intended recipients. 
Increasing the number of model exchanges per aggregation can help reduce the inconsistency among locally aggregated models, but lead to increased communication overhead, especially in large-scale or dynamic networks. 
While some studies, e.g.,~\cite{PINYOANUNTAPONG2022109396,maejima2023tramfl}, have incorporated routing into FL to enhance geographic coverage and reduce model exchange delays, they have not been designed to address the model inconsistency of D-FL resulting from communication errors.

Existing D-FL methods have been typically designed under an implicit assumption of ideal, error-free communication channels, e.g.,~\cite{li2023dfedadmm,9563232,8950073,9562522,10196380}. When applying these methods in practical settings, they would require repeated retransmissions of erroneous (or parts of) models until the models are correctly received, thus causing inefficient use of limited communication resources. 
Some studies have considered the impact of imperfect channels on C-FL (e.g.,~\cite{9726793,9435350,8851249}),  
but the analysis of C-FL cannot be readily applied to D-FL due to their distinctive network protocols and aggregation methods. While a few studies, i.e., \cite{9772390,10032555}, have considered imperfect channels for D-FL, they did not quantify their impact on convergence.
Nonetheless, communication errors can accumulate and increasingly compromise model convergence. 

This paper proposes a new and practical D-FL strategy, namely, Route-and-Aggregate (R\&A) D-FL, where clients exchange their local models with their peers through established routes (as opposed to simple flooding, i.e., AaYG, as generally assumed in the literature).
We also propose to adaptively normalize the aggregation coefficients when a client aggregates its received models, accounting for incorrectly received and subsequently removed segments of the models resulting from imperfect communication links. 
We further analyze the impact of the routing strategy and imperfect communication links on the convergence of R\&A D-FL. The convergence upper bound (or optimality gap) is derived and sheds insights for optimizing the routes for R\&A D-FL with communication errors. 

The contributions of the paper are summarized as~follows.
\begin{itemize}
     \item 
     We design a novel D-FL method, namely, R\&A D-FL, which allows clients to exchange models in a peer-to-peer manner and adaptively normalize the aggregation coefficients when aggregating the models, hence accounting for incorrectly received and subsequently removed (parts of) models resulting from communication errors.
   
   \item We analyze the impact of network topology and communication errors on the convergence of R\&A D-FL. Given the convergence upper bound, we establish the objective for the optimal routing strategy to enhance convergence under practical settings.
    
    \item We reveal that the routing objective can be interpreted as the sum of specified metrics of individual hops along a route, consistent with conventional network routing. Consequently, the optimal routing strategy can be readily derived for R\&A D-FL with communication errors using celebrated network routing algorithms.

\end{itemize}
Extensive experiments validate our convergence analysis and accordingly identify the optimal routing strategy. Three image classification tasks and two next-word prediction tasks are conducted using widely adopted datasets and models.
 In our tested 10-client network, R\&A D-FL outperforms the existing AaYG D-FL~\cite{9563232,8950073} by  35\% in training accuracy. 
In another test with 10 D-FL clients, the training accuracy of R\&A D-FL with transmission errors increasingly approaches that of the ideal C-FL without transmission errors, as the number of routing nodes increases to~28.

The remainder of this paper is organized as follows. The related works are discussed in Section \ref{section: related work}, followed by the system model in Section \ref{section:system}. In Section \ref{section:convergence analysis}, we analyze our convergence upper bound for R\&A D-FL under imperfect communication links and its relationship with the end-to-end (E2E) packet error rate (PER) for the chosen routing paths. The experimental results are provided in Section \ref{section:results}, followed by the conclusions in Section~\ref{section:con}. 
The source code is available at \url{https://github.com/jasminebear2024/RouteAndAggregate}.

\textit{Notation:}  $(\cdot)^\dagger$ stands for transpose; $\Vert\cdot\Vert$ denotes the 2-norm;~$\lceil\cdot\rceil$ denotes the ceiling operator; $|\cdot|$ stands for cardinality; and $\circ$ takes the element-wise product of two vectors or matrices. $\mathbb{E}(\cdot)$ takes expectation over communication errors.
The notation used is listed in Table~\ref{notations}.
\begin{table}[t]
   \small
		\centering
		\caption{Notation and definitions}\label{notations}
		\begin{tabular}{ m{1.1cm}|p{6.8cm}  }
			\hline
			\textbf{Notation}& \textbf{Definition} \\
			\hline
                $m,n$&Index of clients\\
               $\mathcal{V},N$&Set and number of clients, respectively\\
			$i$   & Index of local training epochs\\
                $I$&Maximum number of local training epochs\\
                $t$  &Index of D-FL training rounds\\
                $p_n$&Ideal aggregation coefficient of client $n$\\ $l$&Index of the model parameter segment\\
                ${p}_{m,n,l}$&Aggregation coefficient of the $l$-th segment from client $m$ to client $n$ in R\&A D-FL\\
			$F_n$& Local loss function of client $n$\\
                $F$&Global loss function\\
			$F_n^*$ & Minimum of the local loss function of client $n$\\
                $F^*$&Minimum of the global loss function\\   
                $M$&Dimension of model parameters\\
                $\boldsymbol{\omega}^{t}_{n,i}$ &  Updated local model of client $n$ at the $i$-th epoch\\
                $\boldsymbol{\bar{\omega}}_i^{t}$&Virtually aggregated global model at the $i$-th epoch\\
                 ${\mathbf{W}}_l^{t}$&Set of local training outputs of the $l$-th segment\\
			$h^{t}_{m,n}$&Channel gain from client $m$ to $n$\\
               $\varepsilon_{m,n}^{t}$ & One-hop bit success rate from clients~$m$ to $n$\\
             $\epsilon^t_{m,n}$& One-hop packet success rate from clients~$m$ to~$n$\\
               $\varrho^t_{m,n}$&E2E packet success rate from clients~$m$ to $n$\\
                 ${e}_{m,n,l}^{t}$&Success indicator of the $l$-th model segment from client $m$ to $n$\\
                 $\mathbf{w}_{n}^{t}$&Locally aggregated model of client $n$\\
                $\mathbb{N}_{+}$, $\mathbb{R}_{+}$& Sets of positive integers and real values\\
                \hline
		\end{tabular}
	\end{table}
 
\section{Related Works}\label{section: related work}

Some studies have focused on aggregation schemes for D-FL, typically under the assumption of ideal communication links. 
For example, a fast-linear iteration approach for decentralized averaging under a gossip-based protocol is the most widely used in D-FL to help the locally aggregated parameters approach those of the global model of C-FL~\cite{XIAO200465}.

Some studies, e.g.,~\cite{PINYOANUNTAPONG2022109396,maejima2023tramfl}, have integrated routing into FL to expand geographic coverage and save model exchange delay. 
EdgeML~\cite{PINYOANUNTAPONG2022109396} designed a routing-based C-FL scheme over wireless networks and exploited multi-agent reinforcement learning to minimize network latency to accelerate the convergence of the C-FL. 
In~\cite{maejima2023tramfl}, a routing-based D-FL framework without model transmission and aggregation was proposed called Tram-FL was proposed. Different from the widely considered flooding-type D-FL algorithms, where all clients train local models in parallel, Tram-FL only allows one client to train the global model in each training round. 
When the data is non-i.i.d., the global model of Tram-FL can undergo considerable training performance variation.

Several studies, e.g.,~\cite{9726793,9435350}, have investigated the adverse effects of unreliable communications on C-FL.
In \cite{9726793}, the aggregation coefficients were designed for C-FL to reduce the model bias resulting from transmission errors. 
In~\cite{9435350}, it was shown that the user transmission success probability, affected by unreliable and resource-constrained wireless channels, critically affects the convergence of C-FL. An optimal resource block allocation strategy was proposed to improve transmission success probability and accelerate convergence.

Other studies have designed communication strategies for D-FL, although the impact of communication errors on the convergence of D-FL is typically not explored. The authors of \cite{9716792} considered D-FL over unreliable communications, where the model parameters received with errors were replaced with the receiver's local models.
In~\cite{9772390}, the transmission between two clients was considered unreliable when their communication rate was lower than a certain threshold. 
However, their convergence bound of D-FL was analyzed under the assumption that the transmission errors were negligible. 
In~\cite{10032555}, a user datagram protocol with erasure coding and retransmission was employed to recover content lost during transmissions.
The authors of~\cite{9660377} designed a model compression method to lower the requirements for bandwidth and storage resources of edge clients, and hence improve the system communication efficiency. However, the model compression method can distort the local models, compromising the training accuracy.

In a different yet relevant context, existing studies~\cite{9507294,9292450,10073536,9954055,9945997} have investigated how security threats can make model aggregation unreliable in D-FL systems, but none has analyzed the impacts of unreliable communications. 
In \cite{9507294}, a Byzantine fault-tolerance scheme, combining the HydRand protocol and a publicly verifiable secret-sharing scheme, was developed to protect the confidentiality and robustness of D-FL. 
In~\cite{9292450}, a fault-tolerant blockchain-based D-FL approach was proposed to defend against poisoning attacks. In \cite{10073536}, differential privacy~(DP) noise was added to local model parameters to preserve user privacy. The optimal noise amplitude was designed to balance privacy and utility. FedDual~\cite{9954055} preserves privacy by adding DP noise locally and aggregating asynchronously via a gossip protocol. Noise-cutting was adopted to alleviate the impact of the DP noise on the global model. However, none of these works has addressed unreliable communications.

\section{System Model and Assumptions}\label{section:system}
We consider a system model for the R\&A D-FL method with
 $N$ clients and no central point. Let $\mathcal{V}=\{1,\cdots, N\}$ collect all edge clients. A random topology generated as an undirected graph without self-loops is depicted in Fig. \ref{fig:routing}.  After each participating client completes its local training in a round, its learned local model is delivered to all other clients along the routes specified. 
  Fig. \ref{fig:routing} illustrates the routes along which the local models are delivered to client~1. A local aggregation is conducted at each client in a decentralized manner.
Client $n\in \mathcal{V}$ owns a dataset ${{\cal D}}_n$ with $D_n=|{{\cal D}}_n|$ data points.
 The impact of practical erroneous (wireless) channels on the routing and local model aggregation is captured.

\begin{figure}[t]
            \centering{}
            \includegraphics[scale=0.45]{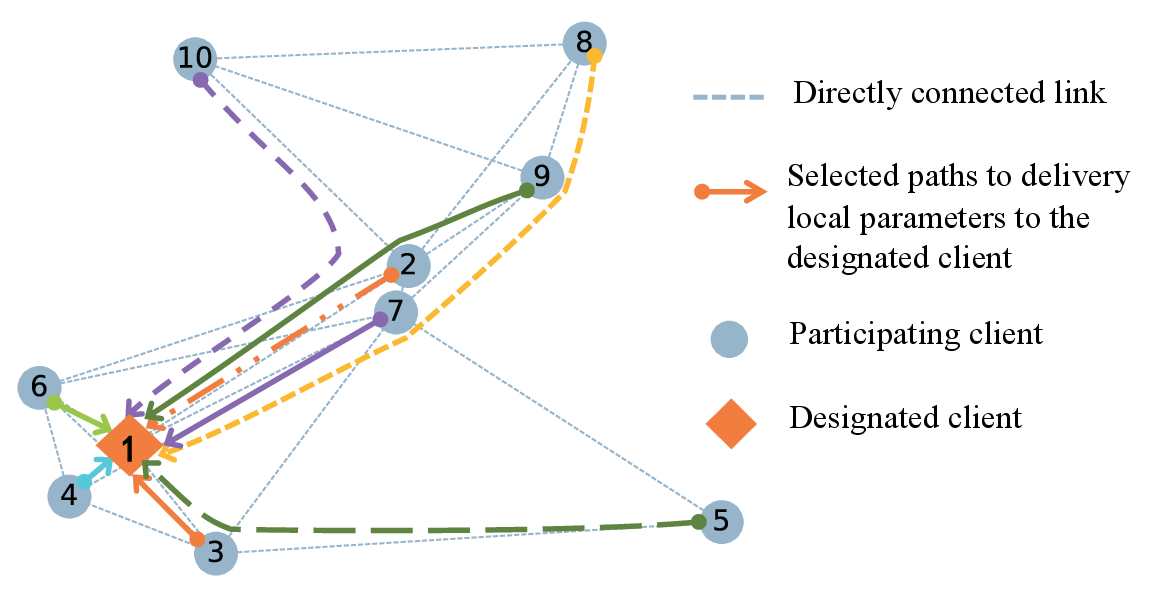}
            \caption{An illustration of the selected routes for local model delivery from participating clients to the designated client in the $t$-th round.}
            \label{fig:routing}
        \end{figure}

\subsection{Channel Model}  
Given the locations of all clients, 
the channel coefficient between every pair of directly connected clients is given, as shown in Fig. \ref{fig:routing}. Suppose the network topology remains stable within a training round. This assumption is reasonable since a typical training round of D-FL lasts tens of seconds~\cite{li2023multi}, during which the network topology is barely changed. The adaptive modulation and coding scheme running on individual links and the adaptive gain controller at the receivers of the links can stabilize the quality of the links, e.g., bit error rate (BER) and PER, during a training round.

Let $h^{t}_{m,n}\in \mathbb{C}$ denote the channel coefficient from client $m$ to client $n$ ($m \neq n$) in the $t$-th training round of D-FL. $h^{t}_{m,n}=h^{t}_{n,m}$, and $h^{t}_{m,n}=0$ if the two clients are beyond each other's transmission coverage. Suppose orthogonal communication channels (e.g., slots) are assigned to the communication links, e.g., by using graph coloring, as will be discussed in Section~\ref{section_allocate_tdma}. Transmission collisions and interference can be prevented among the links.
The signal-to-noise ratio (SNR) is $\gamma^{t}_{m,n}=\frac{(h^{t}_{m,n})^2P}{N_{0}B}$, where $P$ is the transmit power of a client, $N_0$ is the noise power spectral density (PSD), and $B$ is the system bandwidth. 
 With the bit success rate $\varepsilon_{m,n}^{t}$ from client $m$ to $n$, the bit error rate (BER) is $(1-\varepsilon_{m,n}^{t})$.\footnote{
Taking BPSK and QPSK, for example, $1-\varepsilon_{m,n}^{t}=Q(\sqrt{2\gamma^{t}_{m,n}})$~\cite{goldsmith_2005},
where $Q(x)=\frac{1}{\sqrt{2\pi}}\int_{x}^{\infty}e^{-\frac{x^{2}}{2}}dx$ is the error function.
Note that our analysis of the convergence of R\&A D-FL is not limited to a specific channel model.
}

\subsection{D-FL Model}
The local loss function of client $n$ is given by 
 \begin{align}
     F_n(\mathbf{w})=\frac{1}{D_n}{\sum}_{d\in{\cal D}_{n}} {f}(d,\mathbf{w}),
 \end{align}
where $\mathbf{w}\in \mathbb{R}^{M\times 1}$ stands for an $M$-dimensional model parameter, $d \in {{\cal D}}_n$ is a data sample, and ${f}(d,\mathbf{w})$ is the loss function of a data sample $d$ and the model parameter $\mathbf{w}$. 

The objective of D-FL is to minimize the global loss function $\underset{\mathbf{w}\in \mathbb{R}^{M\times 1}}{\min} F(\mathbf{w})$, and find the optimal model parameter $\mathbf{w}^*=\underset {\mathbf{w}\in \mathbb{R}^{M\times 1}}{ \arg\, \min}\,F(\mathbf{w})$, where the global loss function $F(\mathbf{w})$ is the weighted loss of all clients, as given by 
 \begin{align}
      F(\mathbf{w})={\sum}_{n\in{\cal V}} p_nF_n(\mathbf{w}),\label{FW}
 \end{align} 
where $p_n=\frac{D_n}{{\sum}_{n \in \mathcal{V}} D_n}$ is the ideal aggregation coefficient of client~$n$.
Synchronous D-FL is considered. In other words, the training process is divided into rounds synchronized among the clients.
In every training round, each client conducts local training, model delivery, and model aggregation. 

\subsubsection{Local Model Training}
In the $t$-th training round ($t = 1,2,\cdots$), each client $n$ trains $I$ epochs locally based on its local dataset ${\cal D}_n$, starting with its locally aggregated model of the $(t-1)$-th training round. 
The local model, denoted by $\boldsymbol{\omega}^{t}_{n,i}$, is updated by 
\begin{align}\label{epoch_imperfect}
\boldsymbol{\omega}^{t}_{n,i}=\boldsymbol{\omega}^{t}_{n,i-1}-\eta\nabla F_{n}(\boldsymbol{\omega}^{t}_{n,i-1}),\forall i=1,\cdots,I,  
\end{align}
where $\eta$ is the learning rate, $\nabla F_n({\mathbf{\cdot}})$ denotes the (full-batch) gradient of $F_n({\mathbf{\cdot}})$, $\boldsymbol{\omega}^{t}_{n,0}$ is the output of the local aggregation in the $(t-1)$-th training round, and $\boldsymbol{\omega}^{t}_{n,I}$ is the output of client $n$'s local model training in the $t$-th round. 
Although full-batch gradient is considered in this paper, the conclusions drawn can be established for stochastic gradient descent. This is due to the fact that the only difference between full-batch gradient and stochastic gradient descent is the local gradient divergence between the full gradient and the stochastic gradient, as discussed in the proof of \textbf{Lemma~\ref{theo1}} in Appendix~\ref{appendix:theorem proof}.

\subsubsection{Multi-hop Model Delivery}\label{subsection:local aggregation}

Assume that the model parameters of client $n$ are encoded into a bit stream using, e.g., ``float32'' and then segmented into $\lceil \frac{M}{K} \rceil$ packets; i.e., $\boldsymbol{\omega}_{n,I}^{t}=\left[(\boldsymbol{\omega}_{n,I}^{t}(1))^{\dagger},\cdots,(\boldsymbol{\omega}_{n,I}^{t}(\lceil \frac{M}{K} \rceil))^{\dagger}\right]^{\dagger}$, with $K$ elements per packet/segment. $\boldsymbol{\omega}_{n,I}^{t}(l)\in \mathbb{R}^{K\times 1}$. Let $\mathbf{W}_l^{t}\in \mathbb{R}^{N\times K}$ collect the $l$-th segments of the training outputs of all $N$ clients in the $t$-th round, i.e.,
  \begin{align}
\mathbf{W}_l^{t}=\left[\boldsymbol{\omega}^{t}_{1,I}(l),\cdots,\boldsymbol{\omega}^{t}_{N,I}(l)\right]^{\dagger}.\label{aver_model_parameter}%
\end{align}%
Let $\epsilon^t_{m,n}$ be the packet success rate between clients $n$ and $m$, if connected directly in the $t$-th round. The PER is $1-\epsilon^t_{m,n}=1-(\varepsilon_{m,n}^{t})^{32 K}$, with ``$32$'' resulting from ``float32''.
Then, the E2E-PER between any two clients $m'$ and $n'$ ($m'\neq n'$) is 
\begin{align}
    1-\varrho^t_{m',n'}=1-\underset{(m,n)\in S^t_{m',n'}}{\prod}\epsilon_{m,n}^{t},
\end{align}
where $S^t_{m',n'}$ collects the edges along a selected multi-hop path from client $m'$ to client $n'$ in the $t$-th training round.

\subsubsection{Local Model Aggregation}\label{subsubsection:localmodelaggregation}
        
Each client collects its peers' local parameters through the above-mentioned multi-hop transmissions and aggregates the model parameters locally. Assume that the client aggregates the correctly received packets. The aggregation coefficients are adaptively updated for each segment of the model such that the sum of the aggregation coefficients of correctly received packets is normalized for the segment. 
Specifically, the locally aggregated model of client $n$ in the $t$-th training round is $\mathbf{w}_{n}^{t}=\big[(\mathbf{w}_{n}^{t}(1))^{\dagger},\cdots,(\mathbf{w}_{n}^{t}(\lceil \frac{M}{K} \rceil))^{\dagger}\big]^{\dagger}$, where
 \begin{align}
   \mathbf{w}_{n}^{t}(l)=\underset{ m \in \mathcal{V}}{\sum}\frac{p_{m}{e}_{m,n,l}^{t}}{\underset{ m' \in \mathcal{V}}{\sum}p_{m'}{e}_{m',n,l}^{t}}\boldsymbol{\omega}_{m,I}^{t}(l),
   \label{locally_aggregate_model}
 \end{align}
where $p_{m,n,l}^t=\frac{p_{m}{e}_{m,n,l}^{t}}{\underset{ m' \in \mathcal{V}}{\sum}p_{m'}{e}_{m',n,l}^{t}}$ is the aggregation coefficient for the $l$-th segment of client $m$'s local model, if correctly received, at client $n$. Clearly, $\underset{ m \in \mathcal{V}}{\sum}
\frac{p_{m}{e}_{m,n,l}^{t}}{\underset{ m' \in \mathcal{V}}{\sum}p_{m'}{e}_{m',n,l}^{t}}=1$. ${e}^{t}_{m,n,l} $ is the success indicator of the $l$-th segment of the model from client $m$ to client $n$ in the $t$-th round, as given by
\begin{align}\label{error1}
   {e}^{t}_{m,n,l} = \!\!
    \begin{cases}
       \! 1,\!&\!\!\text{if the $l$-th segment is error-free;}\\
       \!  0,\!&\!\!\text{otherwise.}
    \end{cases}
\end{align}
Hence, $\Pr(e_{m,n,l}^{t}=1)=\varrho_{m,n}^{t}$.
$\mathbf{w}_{n}^{t}$ is used to initialize the $(t+1)$-th round of local training at client $n$, i.e., $\boldsymbol{\omega}_{n,0}^{t+1}=\mathbf{w}_{n}^{t}$.

\begin{assumption}\label{assumption}
Some widely used assumptions~\cite{pmlr-v130-ruan21a,boyd2004convex} for analyzing convergence bounds for FL are considered:
	\begin{enumerate}
	    
		\item There exists a constant $L>0$ that the local FL objective $F_n(\mathbf{w}),\forall n$ is $L$-smooth, i.e., 
		$\left\Vert \nabla F_{n}\left(\mathbf{w}_{1}\right)\!-\!\nabla F_{n}\left(\mathbf{w}_{2}\right)\right\Vert\! \leq\! L\left\Vert \mathbf{w}_{1}\!-\!\mathbf{w}_{2}\right\Vert ,\forall \mathbf{w}_{1},\mathbf{w}_{2}\in \mathbb{R}^{M\times 1}$. The global FL objective $F(\mathbf{w})$ is also $L$-smooth.
		\item There exists a constant $\mu>0$ that the local FL objective $F_n(\mathbf{w}),\forall n$ is $\mu$-strongly convex, i.e., $\left\Vert \nabla F_{n}(\mathbf{w}_{1} \!)\!-\!\nabla F_{n}(\mathbf{w}_{2})\right\Vert\!  \!\geq\! \! \mu\left\Vert \mathbf{w}_{1}\!-\!\mathbf{w}_{2}\right\Vert ,\forall \mathbf{w}_{1},\!\mathbf{w}_{2}\!\in \!\mathbb{R}^{M \!\times \! 1}$. The global objective $F(\mathbf{w})$ is $\mu$-strongly~convex.
  \item The learning rate is $0<\eta<\frac{1}{2L}$.
  \item There is an upper bound $\sigma_n^2$ on the gradient divergence between client $n$'s local loss function and the global loss function, i.e., $\Vert\nabla F_{n}(\boldsymbol{\omega})-\nabla F(\boldsymbol{\omega})\Vert^{2}\leq\sigma_{n}^{2},\forall\boldsymbol{\omega},n$. The global divergence is defined as $\bar{\sigma}^2={\sum}_{n\in{\cal V}} p_n \sigma_{n}^{2}$~\cite{10073536}.
	\end{enumerate}
\end{assumption}	

  \section{Convergence Analysis of R\&A D-FL Under Imperfect Links}\label{section:convergence analysis}
In this section, we analyze the convergence of R\&A D-FL under imperfect communication links. 
With reference to C-FL, ideally, we wish the locally aggregated model of the R\&A D-FL, or D-FL in general, in the $t$-th training round to be consistent across all clients:
\begin{align}\label{accurate_model}
\boldsymbol{\bar{\omega}}_I^{t}={\sum}_{ n \in \mathcal{V}}p_{n}\boldsymbol{\omega}_{n,I}^{t}.
\end{align}
This may not be possible due to the impact of the distributed and local model aggregation and the imperfect communication channels between the clients\footnote{
    Note that $\boldsymbol{\bar{\omega}}^{t}$ is different from the global model of C-FL because each round of local training starts with the locally aggregated models of individual clients in D-FL, i.e., $\mathbf{w}_{n}^{t-1}$, as opposed to the centrally aggregated model of the previous training round, i.e., $\boldsymbol{\bar{\omega}}^{t-1}$, in C-FL.}.

A one-round bound of D-FL under the imperfect channel conditions can be established as follows.

\begin{lemma}\label{theo1}
		Under \textbf{Assumption \ref{assumption}} with $L$ and $\mu$ defined therein, the expectation of the distance between 
 the global model of D-FL in the $t$-th training round, i.e., $ \boldsymbol{\bar{\omega}}_I^{t}$,
	and the global optimum of D-FL, i.e., $\mathbf{w}^*$, is bounded as
\begin{align}
  \notag&\!\mathbb{E}\big\{\left\Vert  \boldsymbol{\bar{\omega}}_I^{t}\!\!-\!\!\mathbf{w}^{*}\right\Vert ^{2}\big\}\!\!\leq \! \zeta_1\!\left\Vert  \boldsymbol{\bar{\omega}}_I^{t\!-\!1}\!-\!\mathbf{w}^{*}\right\Vert ^{2}\!\!+\!\zeta_2\bar{\sigma}^2\\\notag
  &+\!\!\zeta_3 \! \Big[\mathbb{E}\big\{\!\big\Vert\!\!\sum_{n\in\mathcal{V}}\!p_{n}\!(\mathbf{w}_{n}^{t\!-\!1}\!\!-\!\!\boldsymbol{\bar{\omega}}_{I}^{t\!-\!1})\!\big\Vert^{2}\!\big\}\!\!+\!\!\eta L\!\!\sum_{n\in\mathcal{V}}\!p_{n}\!\mathbb{E}\!\big\{\!\!\left\Vert \mathbf{w}_{n}^{t\!-\!1}\!\!\!-\!\!\boldsymbol{\bar{\omega}}_{I}^{t\!-\!1}\!\right\Vert^{2}\!\!\big\}\!\Big]\\
		&+\!\!\zeta_4\!\Big[\!\!\sum_{ n\in\mathcal{V}}\!p_{n}\mathbb{E}\big\{\!\Vert\mathbf{w}_{n}^{t\!-\!1}\!\!-\!\!\boldsymbol{\bar{\omega}}_I^{t\!-\!1}\Vert^{2}\big\}\!\!-\!\!\mathbb{E}\!\big\{\!\big\Vert\!\!\sum_{ n\in\mathcal{V}}\!p_{n}\!\mathbf{w}_{n}^{t\!-\!1}\!\!-\!\!\boldsymbol{\bar{\omega}}_I^{t\!-\!1}\big\Vert^{2}\!\big\}\!\Big],\!\label{theorem_expectation}
		\end{align}
where the coefficients on the right-hand side (RHS) of \eqref{theorem_expectation} are 
\begin{align}
   \notag& \zeta_1\!=\!\big(\!1\!\!-\!\!\frac{3\mu\eta}{2}\!+2L\mu\eta^{2}\big)^{I-1}(1+\tau_{\varrho})\!\left(\!1\!-\!2\mu\eta+\eta^{2}L^{2}\right);\\\notag&
   \zeta_2=\!\frac{2(1+\eta)(\!2\eta^{2}L^{2}\!+\!(L\!+\!u)\eta)\left((\!1\!+\!\eta)\!(\!1\!+\!4L^{2}\eta)\right)^{2}}{\!\!1\!+\!4L^{2}+4L^{2}\eta}\\&\notag\quad\times\bigg[\frac{\left((\!1\!+\!\eta)\!(\!1\!+\!4L^{2}\eta)\right)^{I-1}-\left(\!1\!\!-\!\!\frac{3\mu\eta}{2}\!+2L\mu\eta^{2}\right)^{I-1}}{(\!1\!+\!\eta)\!(\!1\!+\!4L^{2}\eta)-(1\!\!-\!\!\frac{3\mu\eta}{2}\!+2L\mu\eta^{2})}\\&\notag\quad\quad\quad\quad-\frac{\left((\!1\!+\!\eta)\!(\!1\!+\!4L^{2}\eta)\right)^{I-1}-1}{(\!1\!+\!\eta)\!(\!1\!+\!4L^{2}\eta)-1}\bigg];\\
   \notag& \zeta_3=\big(\!1\!\!-\!\!\frac{3\mu\eta}{2}\!+2L\mu\eta^{2}\big)^{I-1}\!\!(1\!\!+\!\!\frac{1}{\tau_{\varrho}}\!)\!(1\!\!+\!\!\eta L);\\
   \notag&\zeta_4=\!(\!2\eta^{2}L^{2}\!+\!(L\!+\!u)\eta)\left((\!1\!+\!\eta)\!(\!1\!+\!4L^{2}\eta)\right)^{2}\\\notag&\quad\quad\quad\times\frac{\left((\!1\!+\!\eta)\!(\!1\!+\!4L^{2}\eta)\right)^{I-1}-\left(\!1\!\!-\!\!\frac{3\mu\eta}{2}\!+2L\mu\eta^{2}\right)^{I-1}}{(\!1\!+\!\eta)\!(\!1\!+\!4L^{2}\eta)-(1\!\!-\!\!\frac{3\mu\eta}{2}\!+2L\mu\eta^{2})};
\end{align}
 and $\tau_{\rho}$ indicates the noise level of the clients.%
\end{lemma}%
\begin{proof}%
See \textbf{Appendix \ref{appendix:theorem proof}}.
\end{proof}

Note in \textbf{Lemma~\ref{theo1}} that minimizing the third and fourth terms of \eqref{theorem_expectation}
facilitates the one-round upper bound of D-FL, as the other terms are either constant or independent of communication and aggregation mechanisms. 
Both the third and fourth terms contain the bias between the imperfectly and perfectly aggregated local models, the $l$-th segment of which can be rewritten in matrix forms, as given by
\begin{align}
\big[\boldsymbol{\bar{\omega}}_{I}^{t}(l)-\mathbf{w}_{1}^{t}(l),\cdots,\boldsymbol{\bar{\omega}}_{I}^{t}(l)-\mathbf{w}_{N}^{t}(l) \big]^{\dagger}={\Lambda_l^{t}}\mathbf{W}_l^{t}, 
\end{align}
where ${\Lambda_l^{t}}=\left[\begin{array}{ccc}
\lambda_{1,1,l}^{t} & \cdots & \lambda_{1,N,l}^{t}\\
\vdots & \ddots & \vdots\\
\lambda_{N,1,l}^{t} & \cdots & \lambda_{N,N,l}^{t}
\end{array}\right]$ is the coefficient matrix of the bias with $\lambda_{m,n,l}^{t}=p_{m}-{p}^{t}_{m,n,l}, \forall m,n \in \mathcal{V}$. 
As a result, we can rewrite 
\begin{align}\label{eq:rewritten_bias:a}
\!\Big\Vert{\sum}_{n\in\mathcal{V}}\!p_{n}\!(\mathbf{w}_{n}^{t\!-\!1}\!\!-\!\!\boldsymbol{\bar{\omega}}_{I}^{t\!-\!1})\!\Big\Vert^{2}&\!\!=\!\!{\sum}_{\forall l}\big\Vert\mathbf{p}\Lambda_l^{t\!-\!1}\mathbf{W}_l^{t\!-\!1}\big\Vert^{2};\\
 {\sum}_{n\in\mathcal{V}}\!p_{n}\!\big\Vert\mathbf{w}_{n}^{t\!-\!1}\!\!-\!\!\boldsymbol{\bar{\omega}}_I^{t\!-\!1}\!\big\Vert^{2}&\!\!=\!\!{\sum}_{\forall l}\!\left\Vert \!\mathrm{diag}(\sqrt{\mathbf{p}})\Lambda_l^{t\!-\!1}\mathbf{W}_l^{t\!-\!1}\!\right\Vert^{2}\!\!,\!\label{eq:rewritten_bias:b}
\end{align}
where $\mathbf{p}=[p_1,\cdots,p_n]$. By substituting \eqref{eq:rewritten_bias:a} and \eqref{eq:rewritten_bias:b}, the third term on the RHS of~\eqref{theorem_expectation} 
is upper bounded by
 \begin{subequations}\small\label{rewrite_bias1}%
\begin{align}\notag
& \mathbb{E}\big\{\!\big\Vert\!\!\sum_{n\in\mathcal{V}}\!p_{n}\!(\mathbf{w}_{n}^{t\!-\!1}\!\!-\!\!\boldsymbol{\bar{\omega}}_{I}^{t\!-\!1})\!\big\Vert^{2}\!\big\}\!\!+\!\!\eta L\!\!\sum_{n\in\mathcal{V}}\!p_{n}\!\mathbb{E}\!\big\{\!\left\Vert \mathbf{w}_{n}^{t\!-\!1}\!\!\!-\!\!\boldsymbol{\bar{\omega}}_{I}^{t\!-\!1}\!\right\Vert^{2}\!\!\big\}\\
  \notag =&\mathbb{E}\Big\{\!\sum_{\forall l}\!\big\Vert\mathbf{p}\Lambda_l^{t\!-\!1}\mathbf{W}_l^{t\!-\!1}\big\Vert^{2}\!\Big\}\!\!+\!\!\mathbb{E}\Big\{\!\eta L\!\sum_{\forall l}\left\Vert \mathrm{diag}(\!\sqrt{\mathbf{p}})\Lambda_l^{t\!-\!1}\mathbf{W}_l^{t\!-\!1}\right\Vert ^{2}\!\Big\}\\\label{bound2:a}
\leq&\!\big(\Vert\mathbf{p}\Vert^{2}\!\!+\!\!\eta L\left\Vert \mathrm{diag}(\sqrt{\mathbf{p}})\right\Vert ^{2}\big)\sum_{\forall l}\left(\mathbb{E}\{\Vert\Lambda_l^{t\!-\!1}\Vert^{2}\}\Vert\mathbf{W}_l^{t\!-\!1}\Vert^{2}\right)\\\label{bound2:b}
   \leq&\!\!\left(\!N\Vert \mathrm{diag}(\mathbf{p})\Vert^2\!\!+\!\!\eta L \Vert \mathrm{diag}(\mathbf{p})\Vert\right)\!\!\sum_{\forall l}\!\!\left( \mathbb{E}\{\Vert\!\Lambda_l^{t\!-\!1}\!\Vert^{2}\}\Vert\!\mathbf{W}_l^{t\!-\!1}\!\Vert^{2}\right),
\end{align}%
\end{subequations}%
where \eqref{bound2:a} is due to the homogeneity of 2-norm. \eqref{bound2:b} is based on $\Vert\mathbf{p}\Vert^{2}=\Vert\mathbf{1}_{N}\mathrm{diag}(\mathbf{p})\Vert^{2}\leq N\Vert\mathrm{diag}(\mathbf{p})\Vert^2$ and $\Vert\mathrm{diag}(\sqrt{\mathbf{p}})\Vert^2=\Vert\mathbf{p}\Vert$. 

 The fourth term on the RHS of \eqref{theorem_expectation} 
 is upper bounded by 
\begin{subequations}\small\label{rewrite_bias2}%
\begin{align}
   &{\sum}_{ n\in\mathcal{V}}\!p_{n}\mathbb{E}\big\{\!\Vert\mathbf{w}_{n}^{t\!-\!1}\!\!-\!\!\boldsymbol{\bar{\omega}}_I^{t\!-\!1}\Vert^{2}\big\}\!\!-\!\!\mathbb{E}\!\big\{\!\big\Vert{\sum}_{ n\in\mathcal{V}}\!p_{n}\!\mathbf{w}_{n}^{t\!-\!1}\!\!-\!\!\boldsymbol{\bar{\omega}}_I^{t\!-\!1}\big\Vert^{2}\!\big\}\notag\\
  &\!\!\!\!=\!\!\mathbb{E}\Big\{\!\sum_{\forall l}\left\Vert \mathrm{diag}(\!\sqrt{\mathbf{p}})\Lambda_l^{t\!-\!1}\mathbf{W}_l^{t\!-\!1}\right\Vert ^{2}\!\Big\}\!\!- \!\!\mathbb{E}\Big\{\!\sum_{\forall l}\!\big\Vert\mathbf{p}\Lambda_l^{t\!-\!1}\mathbf{W}_l^{t\!-\!1}\big\Vert^{2}\!\Big\}
  \notag\\
  &\!\!\!\!\leq\!\!\sum_{\forall l}\mathbb{E}\Big\{\left\Vert \mathrm{diag}(\!\sqrt{\mathbf{p}})\Lambda_l^{t\!-\!1}\mathbf{W}_l^{t\!-\!1}\right\Vert ^{2}\!\!- \!\!\big\Vert\mathrm{diag}(\mathbf{p})\Lambda_l^{t\!-\!1}\mathbf{W}_l^{t\!-\!1}\big\Vert^{2}\!\Big\}\label{var_bound:b}
  \\&\!\!\!\!\leq\left\Vert \mathrm{diag}(\sqrt{\mathbf{p}}-\mathbf{p})\right\Vert ^{2}{\sum}_{\forall l}\left( \mathbb{E}\{\Vert\Lambda_l^{t\!-\!1}\Vert^{2}\}\Vert\mathbf{W}_l^{t\!-\!1}\Vert^{2}\right),\label{var_bound:c}%
\end{align}%
\end{subequations}%
where 
\eqref{var_bound:b} is due to $\left\Vert A\right\Vert ^{2}\leq\left\Vert \mathbf{1}_{N}A\right\Vert ^{2}$ with $A=\mathrm{diag}(\mathbf{p})\Lambda_l^{t\!-\!1}\mathbf{W}_l^{t\!-\!1}$; and \eqref{var_bound:c} is due to the homogeneity of $2$-norm.

\begin{lemma}\label{lemma2}By substituting \eqref{rewrite_bias1} and \eqref{rewrite_bias2} into the one-round upper bound in \eqref{theorem_expectation}, the general one-round convergence upper bound of D-FL with respect to the aggregation bias matrix $\Vert\Lambda_l^t\Vert^2$ is given by
\begin{align}
\notag&\!\mathbb{E}\big\{\left\Vert  \boldsymbol{\bar{\omega}}_I^{t}\!\!-\!\!\mathbf{w}^{*}\right\Vert ^{2}\big\}\!\!\leq \! \zeta_1\!\left\Vert  \boldsymbol{\bar{\omega}}_I^{t\!-\!1}\!-\!\mathbf{w}^{*}\right\Vert ^{2}\!\!+\!\zeta_2\bar{\sigma}^2\\\notag
  &+\!\!\big(\zeta_3 N \Vert\mathrm{diag}(\mathbf{p})\Vert^2\!\!+\!\!\zeta_3\eta L \Vert\mathrm{diag}(\mathbf{p})\Vert\!\!+\!\!\zeta_4\Vert\mathrm{diag}(\sqrt{\mathbf{p}}\!\!-\!\!\mathbf{p})\Vert^2\big)\notag\\
     &\quad\quad\times\!{\sum}_{\forall l}\big\Vert \mathbf{W}_l^{t\!-\!1}\big\Vert ^{2}\mathbb{E}\{\Vert\Lambda_l^{t-1}\Vert^2\}.\label{theorem_expectation_2_lambda}%
     \end{align}%
\end{lemma}%
We note that \textbf{Lemmas \ref{theo1}} and \textbf{\ref{lemma2}} establish a generic one-round convergence upper bound for D-FL protocols. Next, we proceed to bound $\mathbb{E}\{\Vert\Lambda_l^{t-1}\Vert^2\}$ for the proposed R\&A D-FL, as summarized in the following lemma.
\begin{lemma}\label{lemma1}
 Under R\&A D-FL, the expectation of the square of $\lambda_{m,n,l}^{t-1},\forall m,n$ in the bias matrix, i.e., $\mathbb{E}\left((\lambda_{m,n,l}^{t-1})^2\right)$, is upper bounded by
    \begin{align}
  \notag\mathbb{E}&\left((\lambda_{m,n,l}^{t-1})^2\right)  \leq(1-\varrho_{m,n}^{t-1})p_m^2
  +\!{\sum}_{ j\in\mathcal{V}_{m}}\!p_{j}(1-\varrho_{j,n}^{t-1})\times\!\!\!\!\!\!\\
  &\;\;\;\;\;\;\;
  \underset{\forall\mathcal{V}'\subseteq\mathcal{V}_{j};m\in\mathcal{V}'}{\sum}\bigg[\frac{p_{m}}{\underset{ k\in\mathcal{V}'}{\sum}p_{k}}\!\underset{ k\in\mathcal{V}'}{\prod}\!\varrho_{k,n}^{t-1}\!\!\!\!\underset{ k'\in\mathcal{V}_{j}\backslash\mathcal{V}'}{\prod}\!\!\!\!\!\!\left(1\!-\!\varrho_{k',n}^{t-1}\right)\!\bigg].\label{eq:bound_of_lambda_nm}
\end{align}
 Here, $\mathcal{V}_{j}\triangleq \mathcal{V}\backslash \{j\},\, \forall j \in \mathcal{V}$.
 
With \eqref{eq:bound_of_lambda_nm}, the expectation of the square of the 2-norm of the bias matrix, i.e., $\mathbb{E}\{\Vert\Lambda_l^{t-1}\Vert^2\}$, is bounded by
    \begin{align}
  &\mathbb{E}\{\Vert\Lambda_l^{t-1}\Vert^2\}\leq{\sum}_{n \in \mathcal{V}}{\sum}_{ m\in \mathcal{V}}\!(1-\varrho_{m,n}^{t-1})\big(p_m^2+p_m\big).\label{eq:bound_of_lambda}
\end{align}
\end{lemma}
\begin{proof}
    See \textbf{Appendix \ref{proof_lemma1}}.
\end{proof}

By substituting \eqref{eq:bound_of_lambda}
into \eqref{theorem_expectation_2_lambda}, 
we arrive at a key finding of this paper, i.e., the one-round upper bound of R\&A D-FL with respect to the E2E-PERs of the routes selected for model delivery, as established in the following theorem.
\begin{theo}\label{theo3}
Given the routes selected between all client pairs, a one-round upper bound on the error of R\&A D-FL is
\begin{align}
\notag&\!\mathbb{E}\big\{\left\Vert  \boldsymbol{\bar{\omega}}_I^{t}\!\!-\!\!\mathbf{w}^{*}\right\Vert ^{2}\big\}\!\!\leq \! \zeta_1\!\left\Vert  \boldsymbol{\bar{\omega}}_I^{t\!-\!1}\!-\!\mathbf{w}^{*}\right\Vert ^{2}\!\!+\!\zeta_2\bar{\sigma}^2\\\notag
  &+\!\!\big(\zeta_3 N \Vert\mathrm{diag}(\mathbf{p})\Vert^2\!\!+\!\!\zeta_3\eta L \Vert\mathrm{diag}(\mathbf{p})\Vert\!\!+\!\!\zeta_4\Vert\mathrm{diag}(\sqrt{\mathbf{p}}\!\!-\!\!\mathbf{p})\Vert^2\big)\notag\\
     &\times\!\sum_{\forall l}\Big\Vert \mathbf{W}_l^{t\!-\!1}\Big\Vert ^{2}\!\!\times\!\!\underset{ n \in \mathcal{V}}{\sum}\underset{ m\in \mathcal{V}}{\sum}\!(1\!\!-\!\!\varrho_{m,n}^{t\!-\!1})\big(p_m^2\!\!+\!\!p_m\big).\label{theorem_expectation2}
	\end{align}
\end{theo}
As unveiled in \textbf{Theorem~\ref{theo3}}, under imperfect channel conditions, 
the one-round convergence upper bound of \textbf{Theorem~\ref{theo3}} is dominated by $\underset{ n \in \mathcal{V}}{\sum}\underset{ m\in \mathcal{V}}{\sum}\!(1-\varrho_{m,n}^{t-1})\big(p_m^2+p_m\big)$ in the last term on the RHS of \eqref{theorem_expectation2}, which depends on the network topology and routing strategy.
Particularly, the one-round convergence upper bound increases monotonically with the E2E-PERs between any pair of clients, i.e., $\{1-\varrho^{t-1}_{m,n},\,\forall m,n\in\mathcal{V}\}$, because $\big(\zeta_3 N \Vert\mathrm{diag}(\mathbf{p})\Vert^2\!+\!\zeta_3\eta L \Vert\mathrm{diag}(\mathbf{p})\Vert\!+\!\zeta_4\Vert\mathrm{diag}(\sqrt{\mathbf{p}}\!-\!\mathbf{p})\Vert^2\big)$ is positive and $\underset{ n \in \mathcal{V}}{\sum}\underset{ m\in \mathcal{V}}{\sum}\!(1-\varrho_{m,n}^{t-1})\big(p_m^2+p_m\big)$ increases with the E2E-PERs, i.e., $1-\varrho^{t-1}_{m,n}, \forall m,n\in\mathcal{V}$.

On the other hand, when the E2E-PERs are small enough, i.e., $1\!-\!\varrho^{t-1}_{m,n}\rightarrow0,\forall m,n\in\mathcal{V}$ and $\tau_{\varrho}\rightarrow0$, the convergence upper bound in \eqref{theorem_expectation2} tends to $\mathbb{E}\big\{\left\Vert  \boldsymbol{\bar{\omega}}_I^{t}\!\!-\!\!\mathbf{w}^{*}\right\Vert ^{2}\big\}\!\!\leq \! \zeta_1\!\left\Vert  \boldsymbol{\bar{\omega}}_I^{t\!-\!1}\!-\!\mathbf{w}^{*}\right\Vert ^{2}\!\!+\!\zeta_2\bar{\sigma}^2$, which is consistent with the one-round convergence upper bound of C-FL in error-free channels. The validity of \eqref{theorem_expectation2} is cross-verified.

\begin{theo}\label{remark 3}
  As $t\rightarrow\infty$, the overall convergence bound of R\&A D-FL can be obtained based on \eqref{theorem_expectation2}.
When $\zeta_1<1$ and the topology and channels dynamically vary per round, this convergence upper bound is given by
 \begin{align}
        \notag&\!\mathbb{E}\big\{\!\!\left\Vert  \boldsymbol{\bar{\omega}}_I^{\infty}\!\!-\!\!\mathbf{w}^{*}\right\Vert ^{2}\!\!\big\}\!\!\leq \! \frac{\zeta_2}{1\!\!-\!\!\zeta_1}\!\bar{\sigma}^2\!+\!\!\sum_{t=1}^{\infty}(\zeta_1)^t\Big[{\underset{ n \in \mathcal{V}}{\sum}\underset{ m\in \mathcal{V}}{\sum}\!(1\!\!-\!\!\varrho^t_{m,n})\big(p_m^2\!\!+\!\!p_m\big)}\Big]\\
  &\times\lambda_{\max}\!\big(\zeta_3 N \Vert\mathrm{diag}(\mathbf{p})\Vert^2\!\!+\!\!\zeta_3\eta L \Vert\mathrm{diag}(\mathbf{p})\Vert\!\!+\!\!\zeta_4\Vert\mathrm{diag}(\sqrt{\mathbf{p}}\!\!-\!\!\mathbf{p})\Vert^2\big) ,\notag
	\end{align}
 where $\lambda_{\max}\geq\sum_{\forall l}\Vert \mathbf{W}_l^{t\!-\!1}\Vert ^{2}, \forall t$ is assumed to be the upper bound on $\sum_{\forall l}\Vert \mathbf{W}_l^{t\!-\!1}\Vert ^{2},\forall t$.
\end{theo}

In light of \textbf{Theorem~\ref{theo3}}, we state another key finding of this paper in the following proposition. While intuitive, the proposition 
provides analytical evidence for strong synergy between R\&A D-FL and networking/routing.
\begin{proposition}\label{theo2}
 Under imperfect channel conditions, the convergence upper bound of R\&A D-FL takes the minimum when the local models are delivered through the routes with the lowest E2E-PERs between the clients. 
\end{proposition}
\begin{proof}
This proposition is readily established based on the discussions so far in this section. 
\end{proof}

In light of \textbf{Proposition \ref{theo2}}, when the network links have abundant bandwidths, 
routing can be decoupled among different pairs of clients. 
We can transform the minimization of the one-round convergence upper bound of R\&A D-FL in \eqref{theorem_expectation2} into $ {N \choose 2}=\frac{1}{2}N(N-1)$ independent subproblems of finding the route with the minimum E2E-PER between each pair of clients. 
This problem is a standard shortest path-finding problem on a weighted graph, where the weight of any edge $(m',n')$ is $-\log(\epsilon_{m',n'}^{t-1})$. 
By finding a route $S^{t-1}_{m,n}$ between clients $m$ and $n$ maximizing $\sum_{{(m',n'){\in S^{t-1}_{m,n}}}}\log(\epsilon_{m',n'}^{t-1})=\log \prod_{{(m',n'){\in S^{t-1}_{m,n}}}}\epsilon_{m',n'}^{t-1}$, we find the route minimizing the E2E-PER, i.e., $1-\prod_{{(m',n'){\in S^{t-1}_{m,n}}}}\epsilon_{m',n'}^{t-1}$. 
The edge weights are positive, and the shortest path-finding problem can be readily solved using the Floyd–Warshall algorithm~\cite{floyd}.

When the bandwidths of the network links are insufficient, routing cannot be decoupled between different pairs of clients.
The routes need to be jointly optimized to minimize
$\underset{ n \in \mathcal{V}}{\sum}\underset{ m\in \mathcal{V}}{\sum}\!(1-\varrho_{m,n}^{t-1})\big(p_m^2+p_m\big)$ or, equivalently, $\underset{m\in\mathcal{V}}{\sum}\left(p_{m}^{2}+p_{m}\right)\underset{n\in\mathcal{V}}{\sum}(1-\varrho_{m,n}^{t-1})$. This is an integer program under the constraints of the maximum transmission time (or communication energy) per client. 
From \eqref{theorem_expectation2}, a client $m$ with more local data points $D_m$ has a larger aggregation coefficient $p_m$ and tends to reduce $\underset{m\in\mathcal{V}}{\sum}\left(p_{m}^{2}+p_{m}\right)\underset{n\in\mathcal{V}}{\sum}(1-\varrho_{m,n}^{t-1})$ to a greater extent. 
Priority should be given to clients with larger $p_m$, when multiple clients compete for limited 
bandwidths to deliver their models.
Suppose $p_1\geq p_2 \geq \cdots \geq p_N$.
We can admit the set of homologous routes (i.e., routes from device $m$ to all other devices) that minimizes $\underset{n\in\mathcal{V}}{\sum}(1-\varrho_{m,n}^{t-1})$ for each device $m$, one device after another, from device $m=1$ to $N$. 

To this end, the routing of R\&A D-FL conforms to standard network routing with computational and communication complexities consistent with those of standard network routing operations. Since network topology discovery, maintenance, and routing are essential for network management and operation, R\&A D-FL does not incur additional complexity for routing compared to standard network systems.

\section{Numerical Results}\label{section:results}
In this section, we conduct extensive experiments to validate the convergence analysis of R\&A D-FL.

\subsection{Experimental Setup}\label{subsection:setup}
Without loss of generality, a 10-node network ($N=10$) is generated at random with a topology of an undirected and connected graph following a random geometric graph model~\cite{penrose2003random}. 
The coordinates of the ten nodes are listed in Table~\ref{tab:coordinate}. 
The connectivity density of edges (i.e., transmitter-receiver pairs) 
is $\rho$; 
the number of directly connected client pairs in the network is $\rho\times\frac{N(N-1)}{2}$.
By default, $\rho=0.5$. 
All clients operate at the radio frequency $f_c=2.5$ GHz with the bandwidth of $B=30$ MHz. The transmit power of each client is $P=20$ dBm. The noise PSD is $N_0=-174$ dBm/Hz. 
We set the channel gain to be $h^t_{m,n} (\text{dB})=20\text{log}(f_c)+20\text{log}(d_{m,n})+32.4$~\cite{rappaport2010wireless}, where $d_{m,n}$ (km) is the distance between clients $m$ and $n$.
The PERs can be evaluated analytically when BPSK or QPSK is applied.

\begin{table}
\small
\centering
\caption{Coordinates ($m$) of 10 randomly generated~clients}
\label{tab:coordinate}
\begin{tabular}{c|c|c|c}
\hline
\textbf{Client ID} & \textbf{Coordinates} & \textbf{Client ID} & \textbf{Coordinates} \\
\hline
1         & (2196, 1351)     &
2         & (3637, 3127)      \\ 
3         & (2642, 284)       &
4         & (2884, 848)       \\ 
5         & (5254, 596)      &
6         & (1730, 1923)      \\ 
7         & (3572, 2668)      &
8         & (4546, 5326)      \\ 
9         & (4328, 4001)      &
10        & (2534, 5171)     \\\hline
\end{tabular}
\end{table}

\subsubsection{Machine Learning Model} We consider the following machine learning (ML) models.
\begin{itemize}
\item \textbf{CNN:}
This model has two convolutional layers with 32 or 64 convolutional filters per layer. We add a pooling layer between the convolutional layers to prevent overfitting. Following the convolutional layers are two fully connected layers. 
We use rectified linear unit (ReLU) in the convolutional and fully connected layers. The model size is $32M=38.72$ Mbits. The learning rate is $\eta=0.001$.
\item \textbf{ResNet18:}
This model consists of 18 layers, including 17 convolutional layers and a fully connected layer. Compared to a plain CNN network, the ResNet network inserts shortcut connections to prevent vanishing gradients~\cite{He_2016_CVPR}.
The size of the ResNet18 model is $32M=374.08$~Mbits. 
The learning rate is~$\eta=0.1$.
\item \textbf{ResNet56:}
This model comprises 56 layers, including 55 convolutional layers and a fully connected layer. It is a smaller version of ResNet18 with a smaller plane number~\cite{He_2016_CVPR}. The size of the ResNet56 model is $32M = 18.92$~Mbits. The learning rate is~$\eta = 0.03$.
\item \textbf{RNN:} Designed for character-level language modeling, this model comprises an embedding layer, a two-layer LSTM with 256 hidden units per layer, and a fully connected output layer. Characters from a vocabulary of 90 are mapped into 8-dimensional vectors, processed by an LSTM, and then passed to the output layer to predict the next character. The size of the RNN model is $32M=27.73$~Mbits. 
The learning rate is~$\eta=0.8$.
\end{itemize}
\subsubsection{Dataset} We consider five widely used public datasets.
\begin{itemize}

\item \textbf{Fed-fashionMNIST:} We consider an image classification task using the CNN to classify the fashion MNIST dataset, which contains $28 \times 28$ grayscale images of 70,000 fashion products from 10 categories (labels), e.g., bags and dresses~\cite{Fashion_mnist}. We regenerate a non-i.i.d. federated fashion MNIST dataset (Fed-fashionMNIST) by dividing the fashion MNIST dataset into 10 non-i.i.d. subsets. Each subset contains image samples of one category. 

\item \textbf{Fed-CIFAR100:} We consider another image classification task to train the ResNet18 model based on the Fed-CIFAR100 dataset. The CIFAR100 dataset comprises 100 classes, with 500 training images and 100 test images per class. Each image has $32\times 32$ pixels. Fed-CIFAR100 
divides the CIFAR100 into 10 training subsets with 10 image classes per subset. Each subset contains the image samples of all 10 classes, a total of 5,000 images.

\item \textbf{CIFAR10:} The ResNet56 model is employed for image classification on the CIFAR10 dataset~\cite{Krizhevsky2009LearningML}, a widely recognized benchmark in ML. CIFAR10 comprises 60,000 color images evenly distributed across 10 classes, with each image having a resolution of $32\times 32$ pixels. For this study, the dataset is divided into 30 non-i.i.d. subsets. 
\item \textbf{Shakespeare (i.i.d.):} For the next-character prediction task, we utilize an RNN model on the Shakespeare dataset, which comprises 715 subsets and is accessed via the TensorFlow Federated API~\cite{pmlr-v54-mcmahan17a}. Each subset includes at least two lines from the speaking role of a character in a play, divided into training lines (the first 80\% of lines for the role) and test lines (the remaining 20\%).

\item \textbf{Shakespeare (non-i.i.d.):} We explore another next-character prediction task on the Shakespeare dataset using the RNN model, where the dataset is partitioned in a non-i.i.d. fashion, as outlined in the FedProx framework~\cite{li2020fedprox}.

\end{itemize}

\subsubsection{Benchmark}The following FL protocols serve as the benchmarks for the proposed R\&A D-FL method.
\begin{itemize}     
    
    \item \textbf{AaYG D-FL~\cite{9716792,9563232,8950073}:}
This is the dominating aggregation method for D-FL, where local models are flooded throughout a network with no routing. 
The existing studies generally assume one local aggregation per FL round~\cite{9563232,8950073}.
For generality, we consider $J$ local aggregations towards the end of each round, i.e., every client aggregates its received models and then broadcasts its aggregated model for $J$ times. $J$ is preconfigurable.

    \item \textbf{C-FL~\cite{pmlr-v54-mcmahan17a}:} 
    A client serves as the aggregator.
    The rest send their local models via routes with the smallest E2E-PERs to the aggregator which aggregates the local models 
    and returns the global model to the clients. The reverse links are also imperfect. A client replaces the erroneous segments of the global model with the corresponding segments of its local model.
    The best aggregator, i.e., node 7, is selected based on training accuracy.
    
\end{itemize}

We also consider two scenarios for local model aggregation, when erroneous local models are received. \footnote{
These methods help prevent biases in model aggregation caused by communication errors. Since their implementation is trivial, both methods incur negligible additional complexity for model aggregation. 
}
\begin{itemize}
    \item \textbf{Adaptive aggregation coefficient normalization:} As proposed in this paper, a client normalizes the sum of the aggregation coefficients for each segment of the local models to be aggregated, where only the models received error-free in that segment are considered; see~\eqref{locally_aggregate_model} and~\eqref{error1}.
    
    \item \textbf{Model substitution~\cite{9716792}:} When any segment of a received local model is erroneous, the recipient replaces that segment with the corresponding segment of its own local model before its local model aggregation.
    
\end{itemize}

\subsubsection{Performance Indicators}\label{section_allocate_tdma}

Apart from training accuracy and loss, we also consider the communication overhead. We use a TDMA protocol to assign time slots for the clients to send their model parameters and consider the broadcast nature of radio transmissions. 
The minimum number of time slots needed can be obtained using edge coloring~\cite{karloff1987efficient}. 
\begin{itemize}
    \item \textbf{R\&A D-FL (Proposed):} 
    Routes are specified before model delivery. 
    To avoid collisions, neighboring clients take different slots. 
    The minimum number of time slots depends on the client that needs the most time slots to accommodate its own and its neighbors' transmissions. 
    
    \item \textbf{AaYG D-FL~\cite{9716792,9563232,8950073}:} 
Every client transmits its (aggregated) local model. The minimum number of slots needed depends on the client with the most neighbors, which is $d_{\max}+1$. Here, $d_{\max}$ is the maximum degree of the network. Consider $J$ transmissions per client in an FL round. $J(d_{\max}+1)$ time slots are required.

    \item \textbf{C-FL~\cite{pmlr-v54-mcmahan17a}:} 
    Similar to R\&A D-FL, routes are given in C-FL. The minimum number of time slots required in C-FL can be specified in the same way as in R\&A D-FL. 
\end{itemize}
\begin{table*}[]\footnotesize\centering
\caption{Total communication overhead per training round}
\label{tab:communication }
\begin{tabular}{cccccccccccccc}
\hline
\multicolumn{2}{c|}{\textbf{Edge Density}} &
  \multicolumn{4}{c|}{\textbf{0.5}} &
  \multicolumn{4}{c|}{\textbf{0.7}} &
  \multicolumn{4}{c}{\textbf{0.9}} \\ \hline
\multicolumn{2}{c|}{\textbf{ML Model}} &
  \multicolumn{1}{c|}{CNN} &
  \multicolumn{1}{c|}{ResNet18} &
   \multicolumn{1}{c|}{ResNet56} &
  \multicolumn{1}{c|}{RNN} &
  \multicolumn{1}{c|}{CNN} &
  \multicolumn{1}{c|}{ResNet18} &
   \multicolumn{1}{c|}{ResNet56} &
  \multicolumn{1}{c|}{RNN} &
  \multicolumn{1}{c|}{CNN} &
  \multicolumn{1}{c|}{ResNet18} &
   \multicolumn{1}{c|}{ResNet56} &
  {RNN} \\ \hline
\multicolumn{1}{c|}{\multirow{2}{*}{R\&A}} &
  \multicolumn{1}{c|}{\ding{172}} &
  \multicolumn{1}{c|}{38} &
  \multicolumn{1}{c|}{38} &
  \multicolumn{1}{c|}{38} &
  \multicolumn{1}{c|}{38} &
  \multicolumn{1}{c|}{40} &
  \multicolumn{1}{c|}{40} &
  \multicolumn{1}{c|}{40} &
  \multicolumn{1}{c|}{40}  &
   \multicolumn{1}{c|}{40} &
  \multicolumn{1}{c|}{40} &
  \multicolumn{1}{c|}{40} &
 40\\ \cline{2-14} 
\multicolumn{1}{c|}{} &
  \multicolumn{1}{c|}{\ding{173}} &
  \multicolumn{1}{c|}{1587} &
  \multicolumn{1}{c|}{15337} &
   \multicolumn{1}{c|}{775} &
  \multicolumn{1}{c|}{1136} &
  \multicolumn{1}{c|}{1587} &
  \multicolumn{1}{c|}{15337} &
   \multicolumn{1}{c|}{775} &
  \multicolumn{1}{c|}{1136} &
  \multicolumn{1}{c|}{1587} &
  \multicolumn{1}{c|}{15337}&
   \multicolumn{1}{c|}{775} &
  1136 \\ \hline
\multicolumn{1}{c|}{\multirow{2}{*}{\begin{tabular}[c]{@{}c@{}}AaYG\\ ($J$=1)\end{tabular}}} &
  \multicolumn{1}{c|}{\ding{172}} &
  \multicolumn{1}{c|}{8} &
  \multicolumn{1}{c|}{8} &
     \multicolumn{1}{c|}{8} &
  \multicolumn{1}{c|}{8} &
  \multicolumn{1}{c|}{10} &
  \multicolumn{1}{c|}{10} &
    \multicolumn{1}{c|}{10} &
  \multicolumn{1}{c|}{10} &
  \multicolumn{1}{c|}{10} &
  \multicolumn{1}{c|}{10}&
  \multicolumn{1}{c|}{10} &
 10 \\ \cline{2-14} 
\multicolumn{1}{c|}{} &
  \multicolumn{1}{c|}{\ding{173}} &
  \multicolumn{1}{c|}{387.2} &
  \multicolumn{1}{c|}{3740.8} &
  \multicolumn{1}{c|}{189.2} &
  \multicolumn{1}{c|}{277.3} &
  \multicolumn{1}{c|}{387.2} &
  \multicolumn{1}{c|}{3740.8} &
  \multicolumn{1}{c|}{189.2} &
  \multicolumn{1}{c|}{277.3} &
  \multicolumn{1}{c|}{387.2} &
 \multicolumn{1}{c|}{3740.8}&
  \multicolumn{1}{c|}{189.2} &
277.3\\ \hline
\multicolumn{1}{c|}{\multirow{2}{*}{\begin{tabular}[c]{@{}c@{}}AaYG\\ ($J$=5)\end{tabular}}} &
  \multicolumn{1}{c|}{\ding{172}} &
  \multicolumn{1}{c|}{40} &
  \multicolumn{1}{c|}{40} &
    \multicolumn{1}{c|}{40} &
  \multicolumn{1}{c|}{40} &
  \multicolumn{1}{c|}{50} &
  \multicolumn{1}{c|}{50} &
    \multicolumn{1}{c|}{50} &
  \multicolumn{1}{c|}{50} &
  \multicolumn{1}{c|}{50} &
 \multicolumn{1}{c|}{50}&
    \multicolumn{1}{c|}{50} &
 50 \\ \cline{2-14} 
\multicolumn{1}{c|}{} &
  \multicolumn{1}{c|}{\ding{173}} &
  \multicolumn{1}{c|}{1936} &
  \multicolumn{1}{c|}{18704} &
     \multicolumn{1}{c|}{946} &
  \multicolumn{1}{c|}{1387} &
  \multicolumn{1}{c|}{1936} &
  \multicolumn{1}{c|}{18704} &
     \multicolumn{1}{c|}{946} &
  \multicolumn{1}{c|}{1387} &
  \multicolumn{1}{c|}{1936} &
   \multicolumn{1}{c|}{18704} &
     \multicolumn{1}{c|}{946} &
  1387\\ \hline
\multicolumn{1}{c|}{\multirow{2}{*}{C-FL}} &
  \multicolumn{1}{c|}{\ding{172}} &
  \multicolumn{1}{c|}{23} &
  \multicolumn{1}{c|}{23} &
    \multicolumn{1}{c|}{23} &
  \multicolumn{1}{c|}{23} &
  \multicolumn{1}{c|}{25} &
  \multicolumn{1}{c|}{25} &
    \multicolumn{1}{c|}{25} &
  \multicolumn{1}{c|}{25} &
  \multicolumn{1}{c|}{25} &
  \multicolumn{1}{c|}{25}  &
    \multicolumn{1}{c|}{25} &
 25\\ \cline{2-14} 
\multicolumn{1}{c|}{} &
  \multicolumn{1}{c|}{\ding{173}} &
  \multicolumn{1}{c|}{1161} &
  \multicolumn{1}{c|}{11222} &
  \multicolumn{1}{c|}{568} &
  \multicolumn{1}{c|}{832} &
  \multicolumn{1}{c|}{1161} &
  \multicolumn{1}{c|}{11222} &
  \multicolumn{1}{c|}{568} &
  \multicolumn{1}{c|}{832} &
  \multicolumn{1}{c|}{1161} &
  \multicolumn{1}{c|}{11222} &
  \multicolumn{1}{c|}{568} &
 832 \\ \hline
\multicolumn{8}{l}{\ding{172}: The minimum number of time slots required per training round;} \\
\multicolumn{8}{l}{\ding{173}: The total network traffic (MBits) per training round.}
\end{tabular}
\end{table*}

\subsection{Experimental Results}
\subsubsection{Comparison of FL protocols and aggregation mechanisms}
\begin{figure}[t]
\centering{}
\includegraphics[scale=0.5]{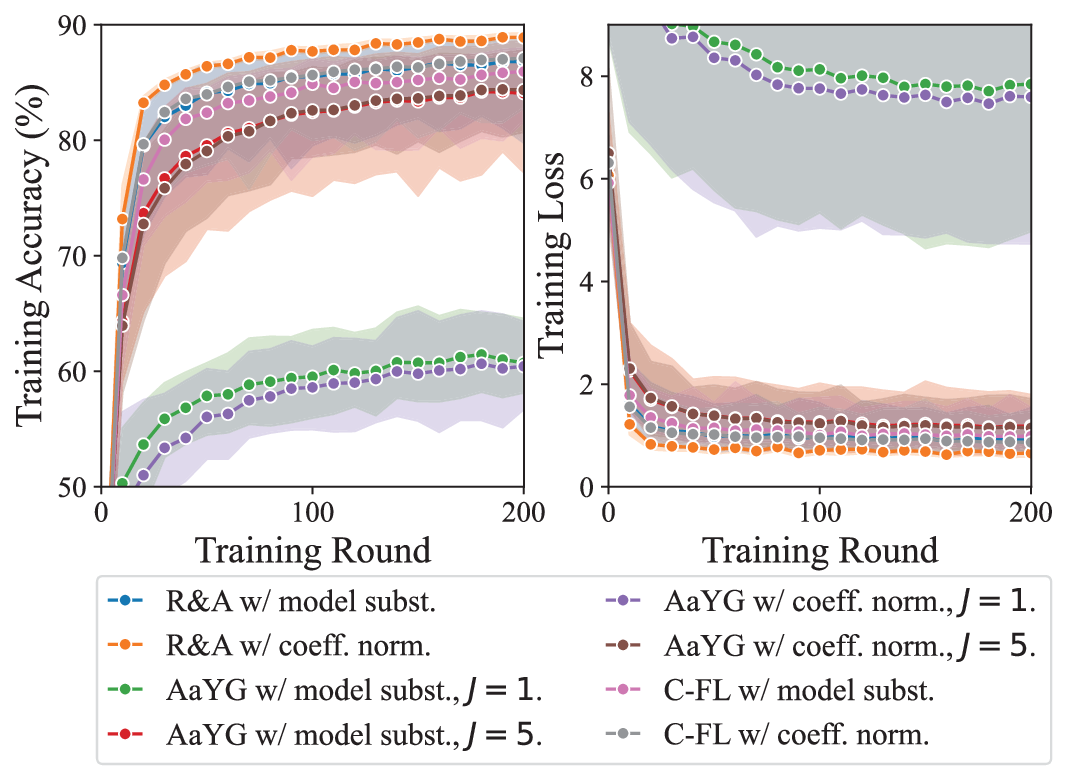}
\caption{The training accuracy and loss of Fed-fashionMNIST dataset on CNN model versus the training round.}
\label{fig:FMNIST_training_acc_curve}
\end{figure}

\begin{figure}[t]
	\centering
	\subfigure[R\&A w/ coeff. norm.]
	{
		\begin{minipage}[t]{0.2\textwidth}
			\centering
				\includegraphics[width=3.9cm]{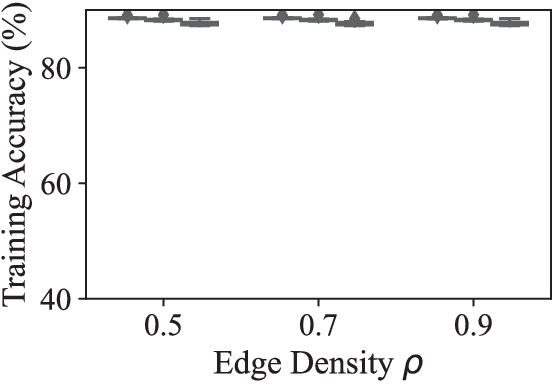}
		\end{minipage}
	}
	\subfigure[R\&A w/ model subst.]
	{
		\begin{minipage}[t]{0.2\textwidth}
			\centering
				\includegraphics[width=3.9cm]{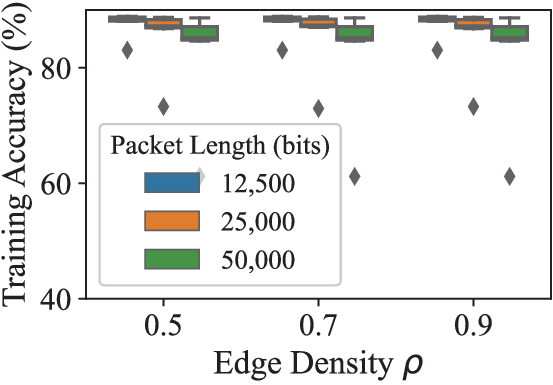}
		\end{minipage}
	} 
 \subfigure[AaYG w/ coeff. norm., $J$=1.]
	{
		\begin{minipage}[t]{0.2\textwidth}
			\centering
				\includegraphics[width=3.9cm]{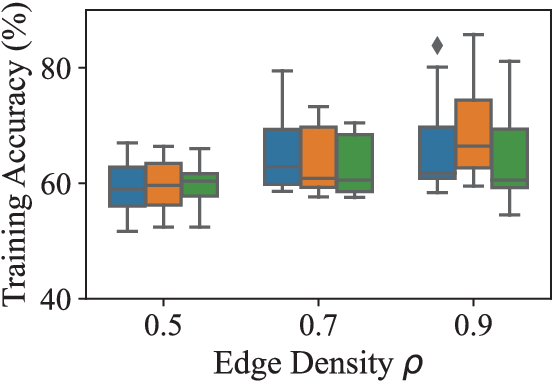}
		\end{minipage}
	}
 \subfigure[AaYG w/ model subst., $J$=1.]
	{
		\begin{minipage}[t]{0.2\textwidth}
			\centering
				\includegraphics[width=3.9cm]{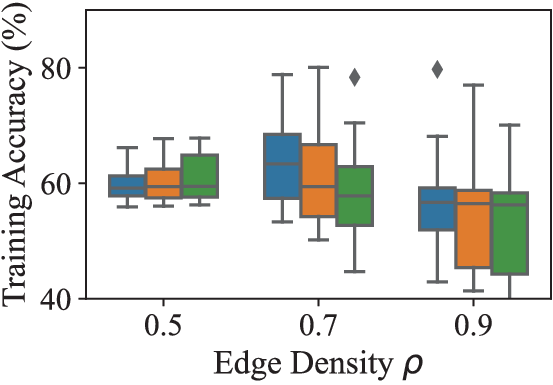}
		\end{minipage}
	}
 \subfigure[AaYG w/ coeff. norm., $J$=5.]
	{
		\begin{minipage}[t]{0.2\textwidth}
			\centering
				\includegraphics[width=3.9cm]{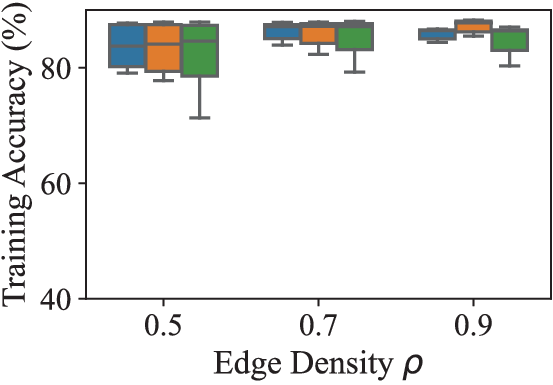}
		\end{minipage}
	}
 \subfigure[AaYG w/ model subst., $J$=5.]
	{
		\begin{minipage}[t]{0.2\textwidth}
			\centering
				\includegraphics[width=3.9cm]{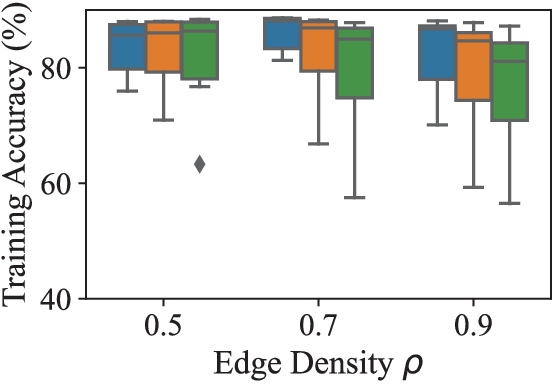}
		\end{minipage}
	}
 \subfigure[C-FL w/ coeff. norm.]
	{
		\begin{minipage}[t]{0.2\textwidth}
			\centering
				\includegraphics[width=3.9cm]{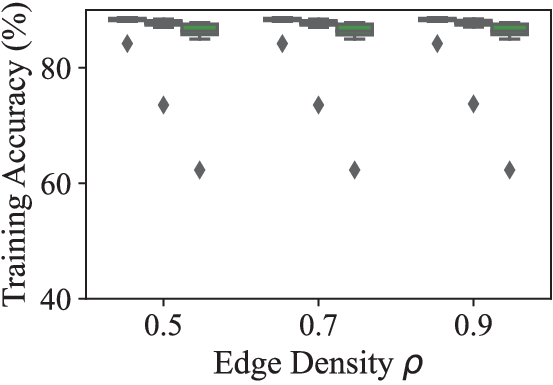}
		\end{minipage}
	}
 \subfigure[C-FL w/ model subst.]
	{
		\begin{minipage}[t]{0.2\textwidth}
			\centering
				\includegraphics[width=3.9cm]{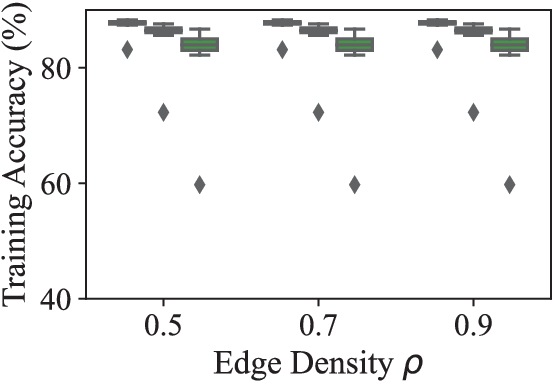}
		\end{minipage}
	}
\caption{Training accuracy of D-FL versus the edge density and the packet length, where the CNN model and non-i.i.d. Fed-FashionMNIST dataset are considered. C-FL selects the best-performing client to serve as the central aggregator.}\label{Convergence curves fashionmnist}
\end{figure}

\begin{figure}[ht]
	\centering
	\subfigure[R\&A w/ coeff. norm.]
	{
		\begin{minipage}[t]{0.2\textwidth}
			\centering
				\includegraphics[width=3.9cm]{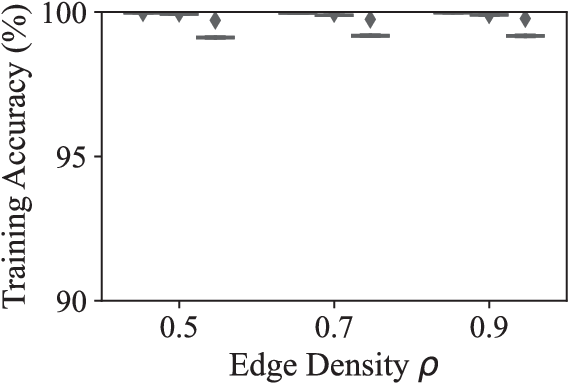}
		\end{minipage}
	}
	\subfigure[R\&A w/ model subst.]
	{
		\begin{minipage}[t]{0.2\textwidth}
			\centering
				\includegraphics[width=3.9cm]{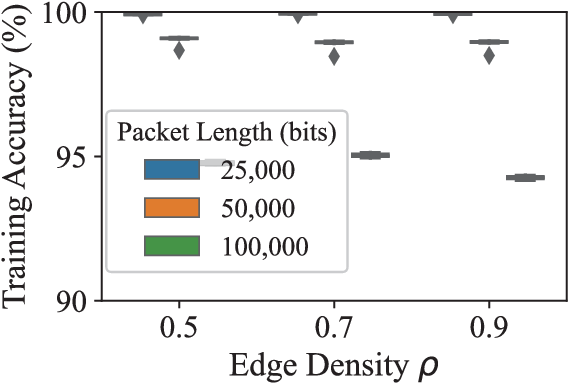}
		\end{minipage}
	} 
 \subfigure[AaYG w/ coeff. norm., $J$=1.]
	{
		\begin{minipage}[t]{0.2\textwidth}
			\centering
				\includegraphics[width=3.9cm]{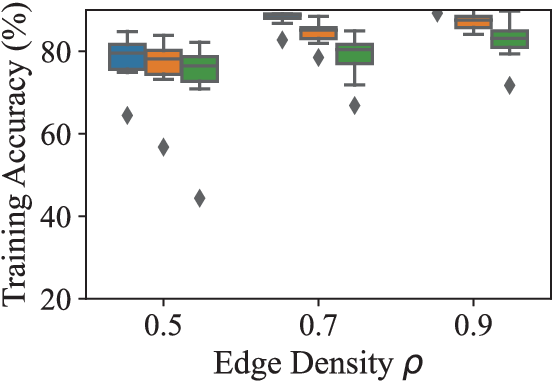}
		\end{minipage}
	}
 \subfigure[AaYG w/ model subst., $J$=1.]
	{
		\begin{minipage}[t]{0.2\textwidth}
			\centering
				\includegraphics[width=3.9cm]{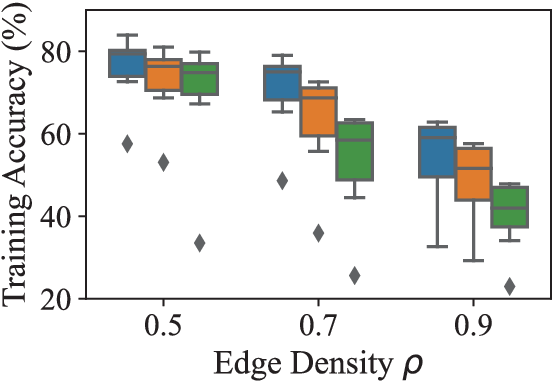}
		\end{minipage}
	}
 \subfigure[AaYG w/ coeff. norm., $J$=5.]
	{
		\begin{minipage}[t]{0.2\textwidth}
			\centering
				\includegraphics[width=3.9cm]{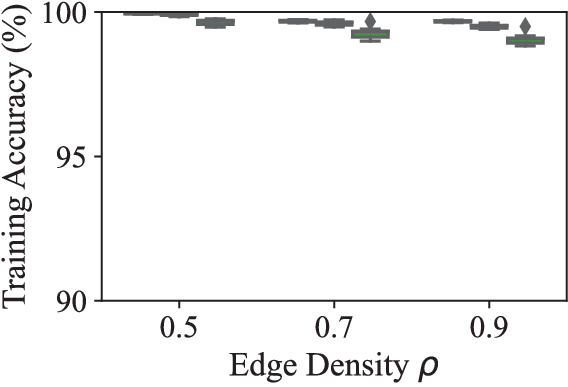}
		\end{minipage}
	}
 \subfigure[AaYG w/ model subst., $J$=5.]
	{
		\begin{minipage}[t]{0.2\textwidth}
			\centering
				\includegraphics[width=3.9cm]{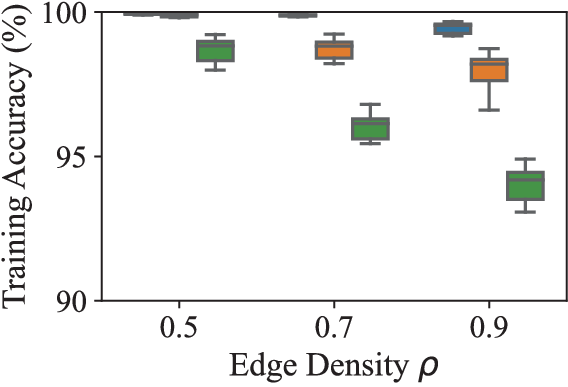}
		\end{minipage}
	}
 \subfigure[C-FL w/ coeff. norm.]
	{
		\begin{minipage}[t]{0.2\textwidth}
			\centering
				\includegraphics[width=3.9cm]{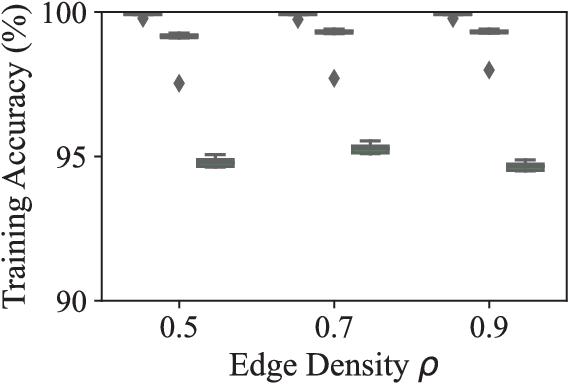}
		\end{minipage}
	}
 \subfigure[C-FL w/ model subst.]
	{
		\begin{minipage}[t]{0.2\textwidth}
			\centering
				\includegraphics[width=3.9cm]{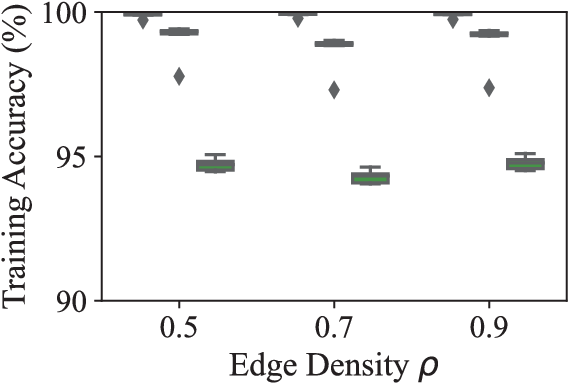}
		\end{minipage}
	}
\caption{Training accuracy vs. training rounds, where the ResNet18 model and non-i.i.d. Fed-CIFAR100 dataset are considered. C-FL selects the best-performing client to serve as the central aggregator.}\label{Convergence curves cifar}
\end{figure}
\begin{figure}[ht]
	\centering
	\subfigure[R\&A w/ coeff. norm.]
	{
		\begin{minipage}[t]{0.2\textwidth}
			\centering
				\includegraphics[width=3.9cm]{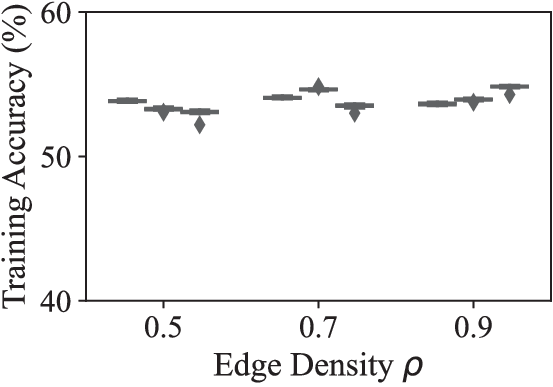}
		\end{minipage}
	}
	\subfigure[R\&A w/ model subst.]
	{
		\begin{minipage}[t]{0.2\textwidth}
			\centering
				\includegraphics[width=3.9cm]{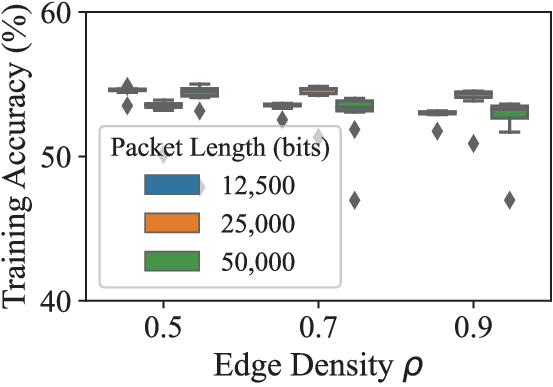}
		\end{minipage}
	} 
 \subfigure[AaYG w/ coeff. norm., $J$=1.]
	{
		\begin{minipage}[t]{0.2\textwidth}
			\centering
				\includegraphics[width=3.9cm]{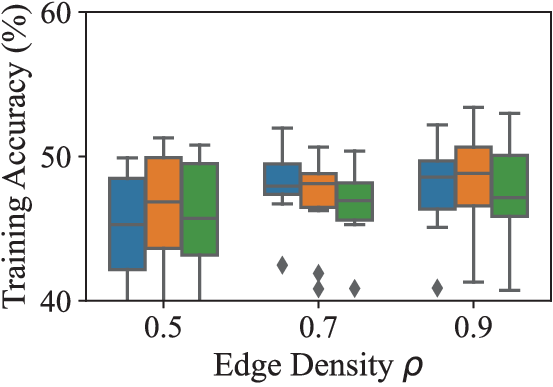}
		\end{minipage}
	}
 \subfigure[AaYG w/ model subst., $J$=1.]
	{
		\begin{minipage}[t]{0.2\textwidth}
			\centering
				\includegraphics[width=3.9cm]{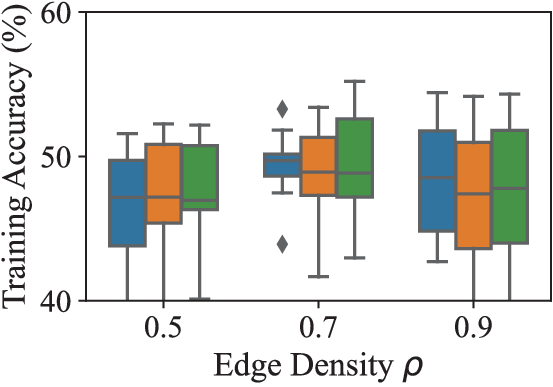}
		\end{minipage}
	}
 \subfigure[AaYG w/ coeff. norm., $J$=5.]
	{
		\begin{minipage}[t]{0.2\textwidth}
			\centering
				\includegraphics[width=3.9cm]{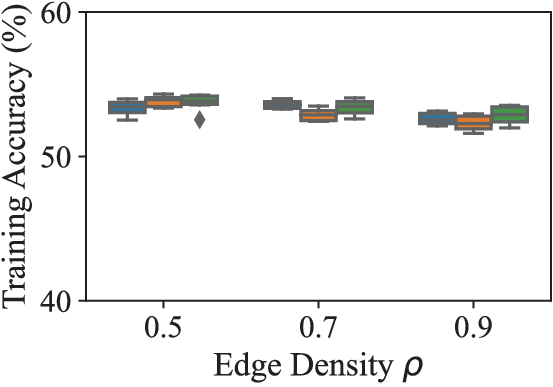}
		\end{minipage}
	}
 \subfigure[AaYG w/ model subst., $J$=5.]
	{
		\begin{minipage}[t]{0.2\textwidth}
			\centering
				\includegraphics[width=3.9cm]{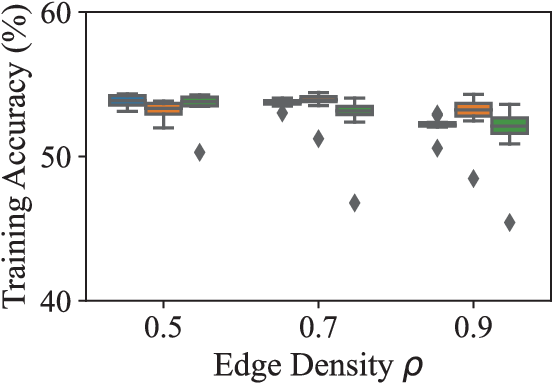}
		\end{minipage}
	}
 \subfigure[C-FL w/ coeff. norm.]
	{
		\begin{minipage}[t]{0.2\textwidth}
			\centering
				\includegraphics[width=3.9cm]{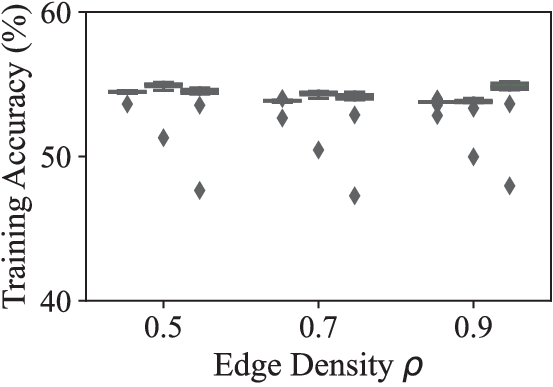}
		\end{minipage}
	}
 \subfigure[C-FL w/ model subst.]
	{
		\begin{minipage}[t]{0.2\textwidth}
			\centering
				\includegraphics[width=3.9cm]{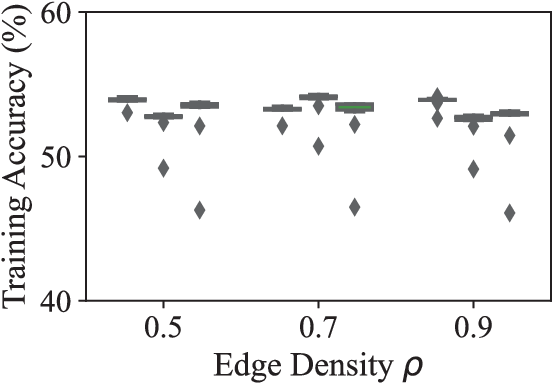}
		\end{minipage}
	}
\caption{Training accuracy vs. training rounds, where the ResNet56 model and CIFAR10 dataset are considered. C-FL selects the best-performing client to serve as the central aggregator.}\label{Convergence curves CIFAR10}
\end{figure}
\begin{figure}[ht]
	\centering
	\subfigure[R\&A w/ coeff. norm.]
	{
		\begin{minipage}[t]{0.2\textwidth}
			\centering
				\includegraphics[width=3.9cm]{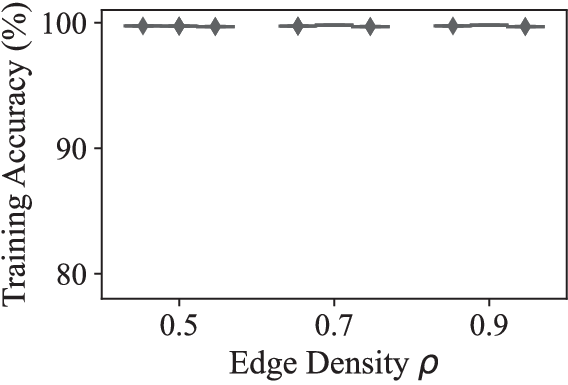}
		\end{minipage}
	}
	\subfigure[R\&A w/ model subst.]
	{
		\begin{minipage}[t]{0.2\textwidth}
			\centering
				\includegraphics[width=3.9cm]{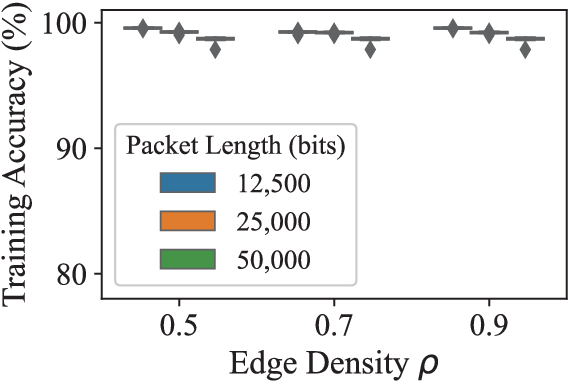}
		\end{minipage}
	} 
 \subfigure[AaYG w/ coeff. norm., $J$=1.]
	{
		\begin{minipage}[t]{0.2\textwidth}
			\centering
				\includegraphics[width=3.9cm]{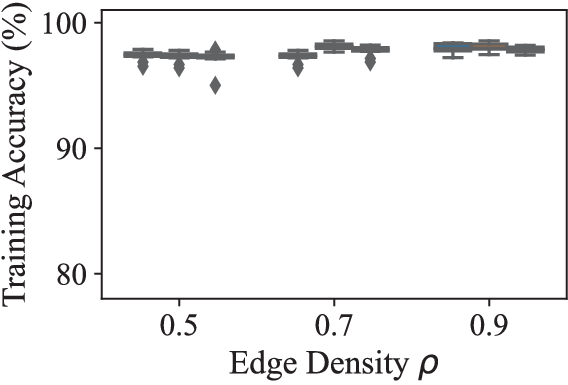}
		\end{minipage}
	}
 \subfigure[AaYG w/ model subst., $J$=1.]
	{
		\begin{minipage}[t]{0.2\textwidth}
			\centering
				\includegraphics[width=3.9cm]{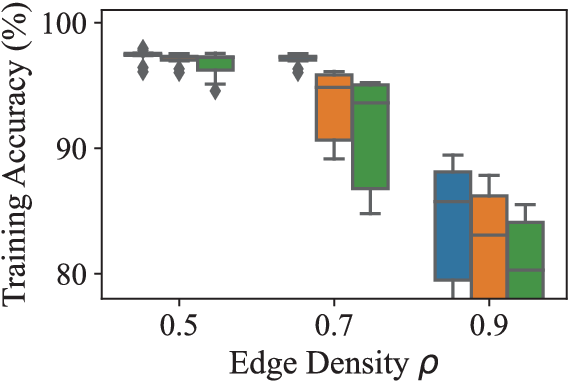}
		\end{minipage}
	}
 \subfigure[AaYG w/ coeff. norm., $J$=5.]
	{
		\begin{minipage}[t]{0.2\textwidth}
			\centering
				\includegraphics[width=3.9cm]{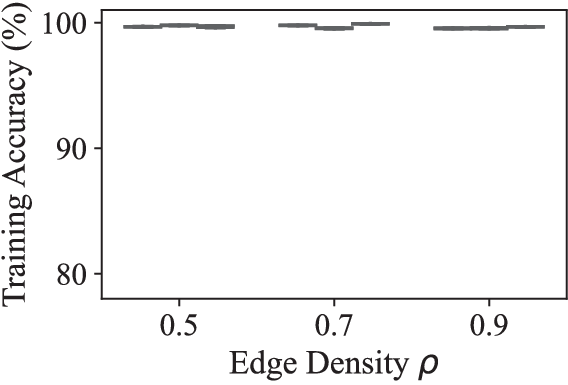}
		\end{minipage}
	}
 \subfigure[AaYG w/ model subst., $J$=5.]
	{
		\begin{minipage}[t]{0.2\textwidth}
			\centering
				\includegraphics[width=3.9cm]{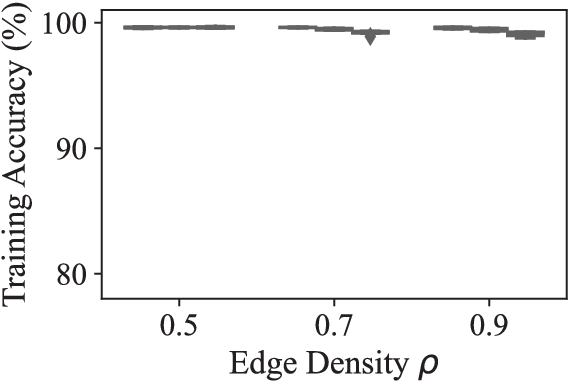}
		\end{minipage}
	}
 \subfigure[C-FL w/ coeff. norm.]
	{
		\begin{minipage}[t]{0.2\textwidth}
			\centering
				\includegraphics[width=3.9cm]{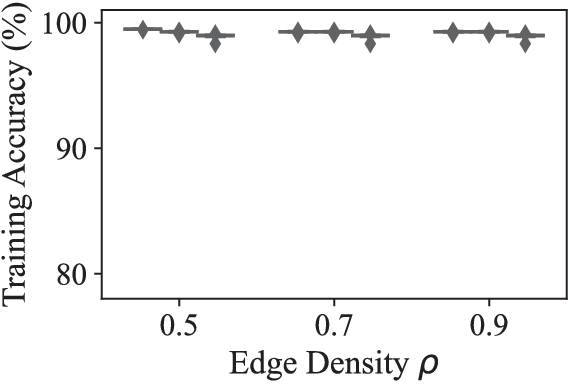}
		\end{minipage}
	}
 \subfigure[C-FL w/ model subst.]
	{
		\begin{minipage}[t]{0.2\textwidth}
			\centering
				\includegraphics[width=3.9cm]{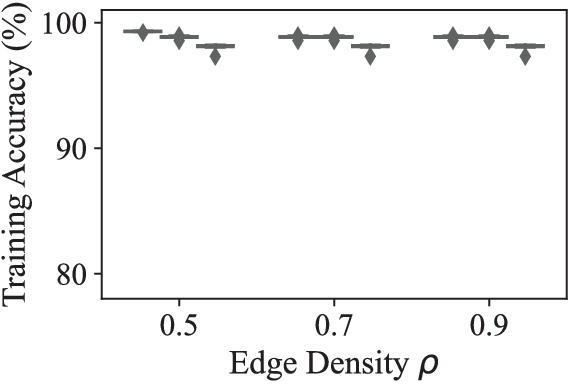}
		\end{minipage}
	}
\caption{Training accuracy vs. training rounds, where the RNN model and Shakespeare (i.i.d.) dataset are considered. C-FL selects the best-performing client to serve as the central aggregator.}\label{Convergence curves original Shakespeare}
\end{figure}
\begin{figure}[ht]
	\centering
	\subfigure[R\&A w/ coeff. norm.]
	{
		\begin{minipage}[t]{0.2\textwidth}
			\centering
				\includegraphics[width=3.9cm]{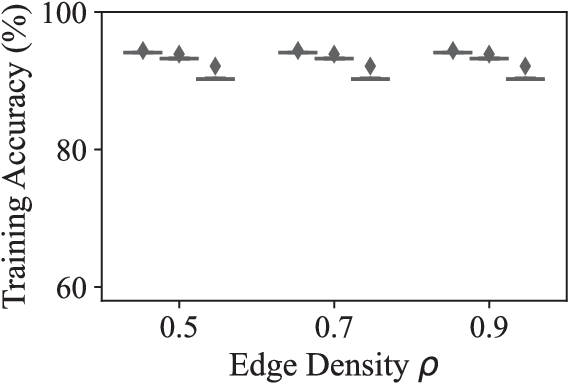}
		\end{minipage}
	}
	\subfigure[R\&A w/ model subst.]
	{
		\begin{minipage}[t]{0.2\textwidth}
			\centering
				\includegraphics[width=3.9cm]{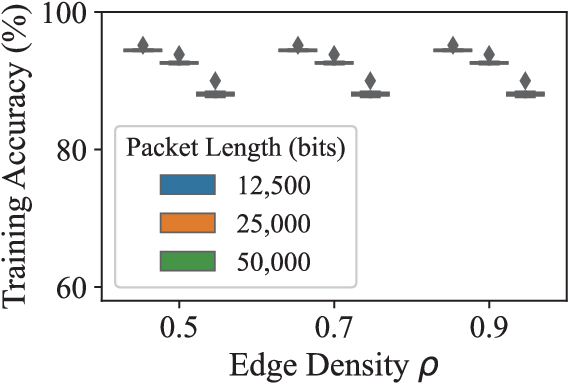}
		\end{minipage}
	} 
 \subfigure[AaYG w/ coeff. norm., $J$=1.]
	{
		\begin{minipage}[t]{0.2\textwidth}
			\centering
				\includegraphics[width=3.9cm]{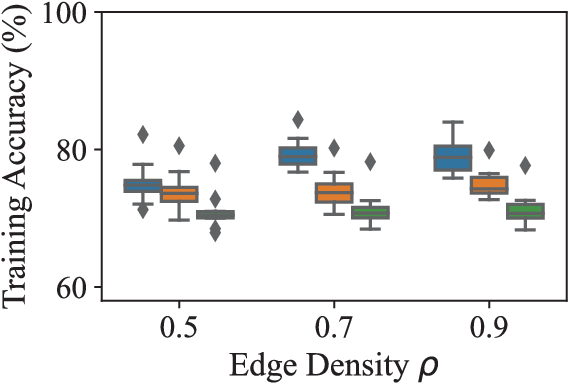}
		\end{minipage}
	}
 \subfigure[AaYG w/ model subst., $J$=1.]
	{
		\begin{minipage}[t]{0.2\textwidth}
			\centering
				\includegraphics[width=3.9cm]{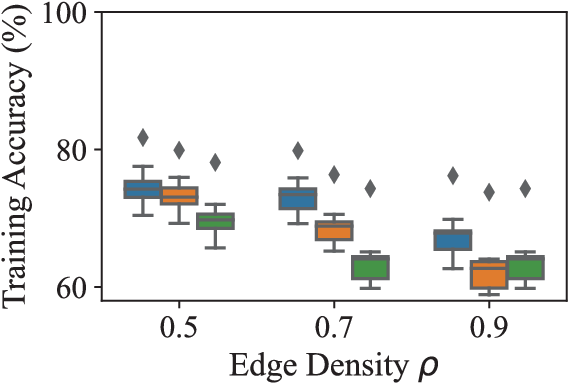}
		\end{minipage}
	}
 \subfigure[AaYG w/ coeff. norm., $J$=5.]
	{
		\begin{minipage}[t]{0.2\textwidth}
			\centering
				\includegraphics[width=3.9cm]{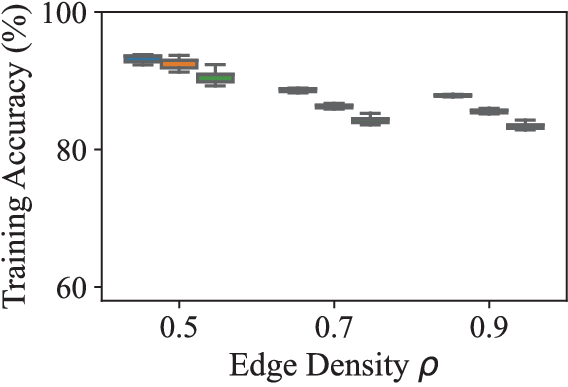}
		\end{minipage}
	}
 \subfigure[AaYG w/ model subst., $J$=5.]
	{
		\begin{minipage}[t]{0.2\textwidth}
			\centering
				\includegraphics[width=3.9cm]{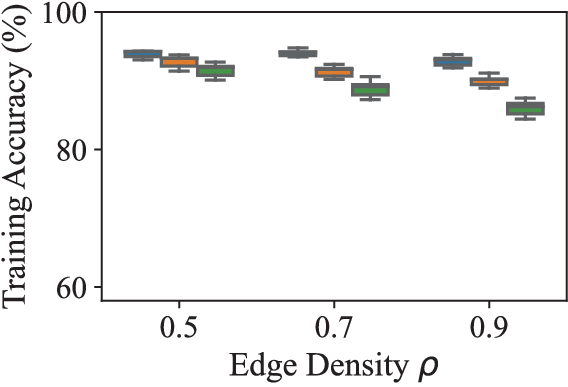}
		\end{minipage}
	}
 \subfigure[C-FL w/ coeff. norm.]
	{
		\begin{minipage}[t]{0.2\textwidth}
			\centering
				\includegraphics[width=3.9cm]{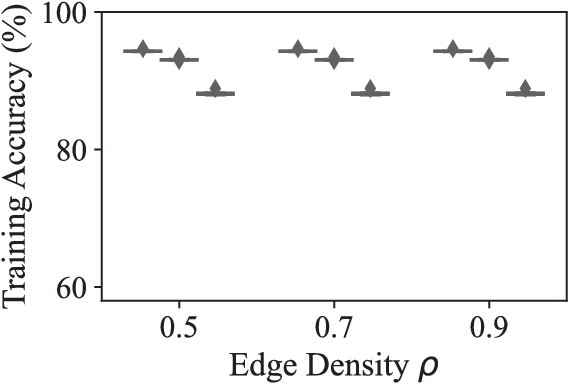}
		\end{minipage}
	}
 \subfigure[C-FL w/ model subst.]
	{
		\begin{minipage}[t]{0.2\textwidth}
			\centering
				\includegraphics[width=3.9cm]{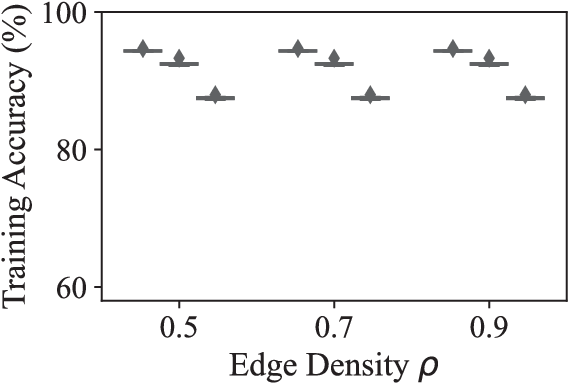}
		\end{minipage}
	}
\caption{Training accuracy vs. training rounds, where the RNN model and Shakespeare (non-i.i.d.) dataset are considered. C-FL selects the best-performing client to serve as the central aggregator.}\label{Convergence curves Shakespeare}
\end{figure}
\begin{figure}
    \centering
    \includegraphics[scale=0.5]{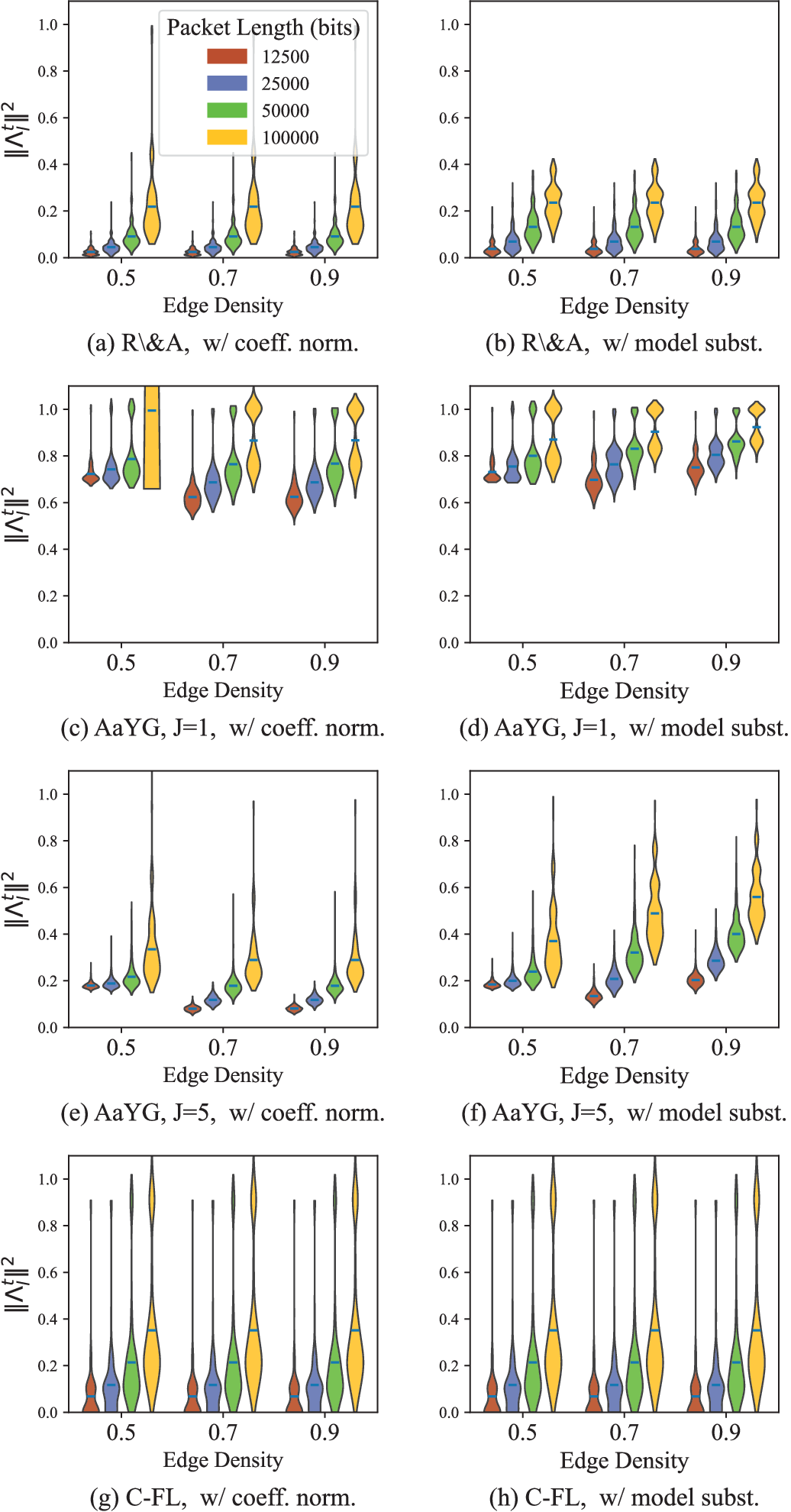}
    \caption{The distribution and means of  $\Vert\Lambda_l^t\Vert^2$ versus the edge density and packet length under different D-FL aggregation schemes.
    }
    \label{fig:2norm_lambda}
\end{figure}

Fig. \ref{fig:FMNIST_training_acc_curve} plots the training accuracy of the CNN model on the Fed-FashionMNIST dataset under different D-FL protocols, where the packet length is $25,000$ bits
and the number of FL rounds is $200$. 
Each solid line plots the average training accuracy of the ten clients under an FL protocol (i.e., R\&A D-FL, AaYG D-FL, and C-FL). Its corresponding light-color shadow presents the distribution of the individual training accuracy for the ten clients. 

It is observed in Fig. \ref{fig:FMNIST_training_acc_curve}  that the proposed R\&A D-FL with adaptive aggregation coefficient normalization can substantially outperform the rest of the schemes in both convergence and accuracy. Not only is the average accuracy of the proposed scheme is higher, e.g., by about 3\% (89\% vs. 86\% under C-FL with adaptive aggregation coefficient normalization), but the clients tend to achieve more consistent training accuracy (with a smaller shadow area) than the other schemes.
It is also observed in Fig. \ref{fig:FMNIST_training_acc_curve} that the use of adaptive aggregation coefficient normalization contributes to the training accuracy and convergence under R\&A D-FL, compared to the model substitution. This is because, 
under model substitution, the local aggregations of a distant and poorly connected client are likely to be dominated by its own local model, and therefore,
the model consistency is penalized.

Figs. \ref{Convergence curves fashionmnist} to~\ref{Convergence curves Shakespeare} plot the training accuracy upon convergence for the considered image classification and next character prediction tasks, where different packet lengths and edge densities are considered. The training accuracy plotted is the average of the last ten FL rounds under each considered setting. The communication overheads of the considered FL protocols are compared in Table~\ref{tab:communication }. 
It is observed in Figs.~\ref{Convergence curves fashionmnist} to \ref{Convergence curves Shakespeare} that R\&A D-FL with adaptive aggregation coefficient normalization is superior to the rest of the schemes in training accuracy (both the average training accuracy and the consistency between individual clients). 
A quantitative comparison of communication overhead  is provided in Table~\ref{tab:communication }.
Under different network densities, we summarize the the minimum number of time slots required for each ML model, and the associated total network traffic (in MBits) per round. 
The communication overhead is substantially lower under R\&A D-FL than under the existing AaYG D-FL (with $J=5$ per round). Although the AaYG D-FL consumes less communication overhead when $J<5$, its training accuracy is penalized.
Additionally, R\&A D-FL offers significantly better learning performance than C-FL despite its slightly higher communication overhead than that of C-FL, as shown in Figs.~\ref{Convergence curves fashionmnist} to \ref{Convergence curves Shakespeare}. 

It is noted in Figs. \ref{Convergence curves fashionmnist} to \ref{Convergence curves Shakespeare} that adaptive aggregation coefficient normalization suffers substantially less from the inconsistency among the models trained and aggregated in a distributed manner, i.e., in the two D-FL protocols, compared to model substitution. Specifically, one or a few clients can produce significantly different training accuracy from the rest of the clients, when model substitution is adopted. The reason is that a client far from the rest has its local aggregations dominated by its own local model, thus distancing itself from the other clients upon convergence, as also noticed in Fig.~\ref{fig:FMNIST_training_acc_curve}.

Fig. \ref{fig:2norm_lambda} plots the distribution and mean of the $2$-norm of the aggregation bias matrix, i.e., $\Vert\Lambda_l^t\Vert^2$, under different D-FL protocols. Here, $\Vert\Lambda_l^t\Vert^2$ captures the aggregation bias between the ideal global model and the imperfectly aggregated local models. $\Vert\Lambda_l^t\Vert^2$ is inherently random since the channels are imperfect. The training accuracy in Figs. \ref{Convergence curves fashionmnist} to \ref{Convergence curves Shakespeare} are consistent with the means of $\Vert\Lambda_l^t\Vert$ in Fig. \ref{fig:2norm_lambda}. Moreover, a smaller mean of $\Vert\Lambda_l^t\Vert^2$ results in better training accuracy and smaller differences in training accuracy among different clients.

\subsubsection{Impact of routing-only nodes}
We evaluate R\&A D-FL in the scenario where there are routing nodes in addition to the ten participating FL clients in the network. The scale of the network is expanded twice, both horizontally and vertically. The routing nodes only relay the local models and do not participate in model training. 
Fig. \ref{fig:forwarding_only} plots the training accuracy of the Fed-fashionMNIST dataset and CNN model under different numbers of routing-only nodes. It is observed that adding more routing nodes can improve the training accuracy and model consistency because the increasing number of routing nodes helps improve the routing diversity of the networks, improving the E2E-PERs. 
When the density of routing nodes is high enough, the E2E-PERs of identified paths can be negligible. As a result, R\&A D-FL and C-FL can provide similar training accuracy, both approaching C-FL obtained with error-free delivery of local and global models. 

\begin{figure}[t]
    \centering
    \includegraphics[scale=0.5]{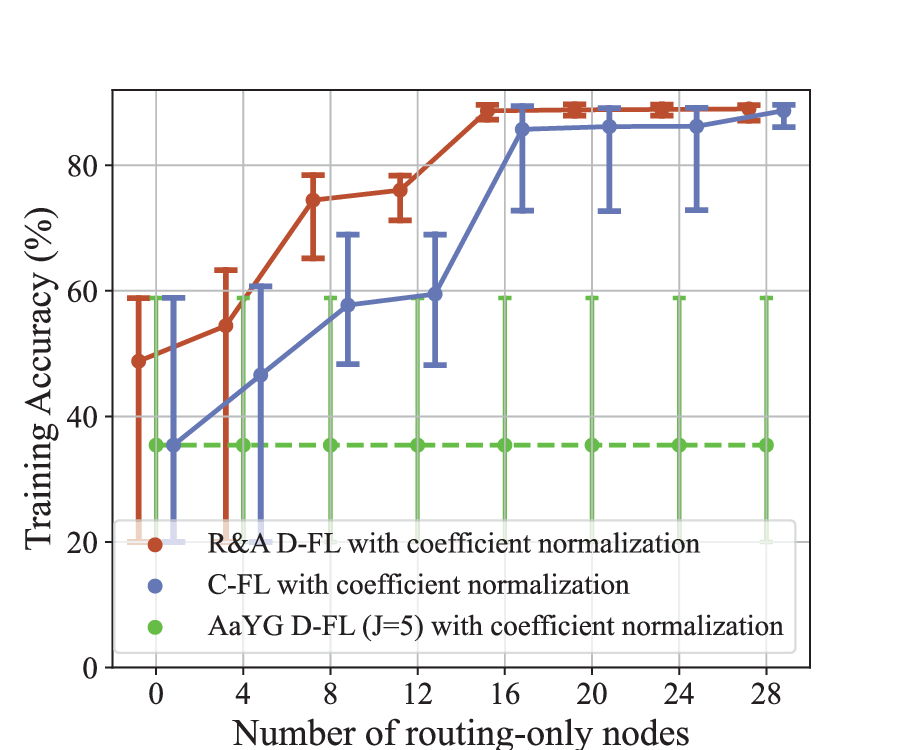}
    \caption{The training accuracy vs. the number of routing-only nodes, where the Fed-fashionMNIST dataset and CNN model are considered. R\&A D-FL is compared with C-FL since AaYG D-FL is unaffected by the number of routing nodes. }
    \label{fig:forwarding_only}
\end{figure}

\subsubsection{Impact of network topology}
We examine the impact of network topology on the local aggregation of R\&A D-FL. This is done by generating 10,000 independent realizations of the instantaneous reception qualities and subsequent 10,000 sets of aggregation coefficients in the aforementioned 10-node network, and plotting the distributions of the aggregation coefficients at every client. 
Fig. \ref{Fig_coefficient_Distribusion} plots the statistical distributions of the aggregation coefficients at each of the ten clients. It is noticed that the aggregation coefficient of a client's model at another client depends heavily on the E2E-PER between the two clients. The larger the E2E-PER, the more dramatically the aggregation coefficient varies.
Moreover, the higher E2E-PERs a client has compared to the other clients (in other words, the client is distant from the rest of the clients, e.g., node 5), the larger aggregation coefficient the client applies to its own local model. The contributions of the others' models diminish. 

\begin{figure}[t]
    \centering
    \includegraphics[width=8cm]{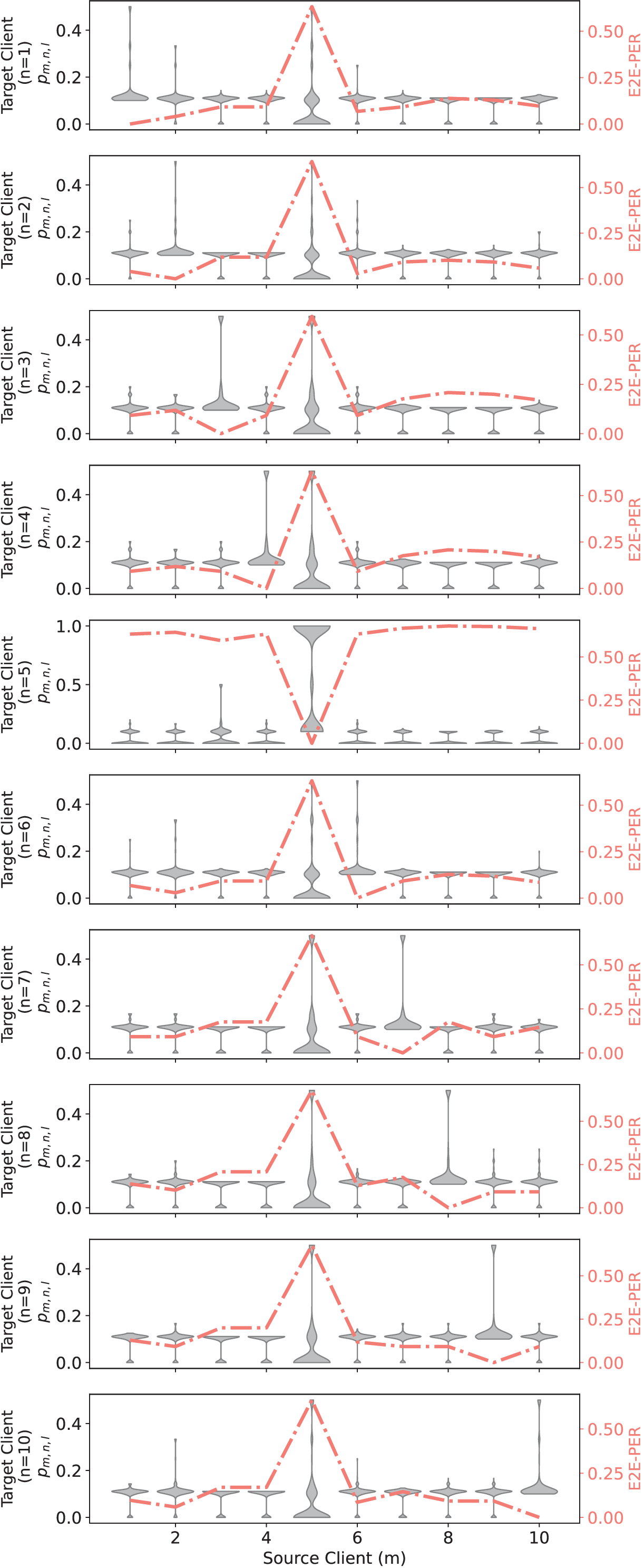}
    \caption{Distribution of the aggregation coefficients, i.e., $p_{m,n,l}=\frac{p_{m}\mathbf{e}_{m,n,l}^{t-1}}{\underset{ m' \in \mathcal{V}}{\sum}p_{m'}\mathbf{e}_{m',n,l}^{t-1}}$; and the E2E-PER, i.e., $1\!-\!\varrho_{m,n}^{t-1}$, from the source client $m$ to the target client $n$ with different packet lengths.}
    \label{Fig_coefficient_Distribusion}
\end{figure}

 \section{Conclusion}\label{section:con}
We have proposed an approach for distributed federated learning (D-FL), namely, R\&A D-FL, which allows D-FL clients to exchange their models with peers and normalize their aggregation coefficients to alleviate the adverse effect of model delivery errors in a peer-to-peer (or device-to-device) communication scenario. 
We have analyzed the impact of the routing strategy and imperfect communication links on the convergence of R\&A D-FL, finding that the optimal routing strategy is to minimize the E2E-PERs between every pair of clients.
Experiments using a CNN and ResNet18 on the Fed-fashionMNIST and Fed-CIFAR100 datasets have validated our analysis and findings. R\&A D-FL surpasses AaYG D-FL in training accuracy by 35\% in the considered network with ten clients, establishing strong synergy between D-FL and networking.
In another test with ten D-FL clients, the training accuracy of R\&A D-FL with transmission errors increasingly approaches that of the ideal C-FL without transmission errors, as the number of routing nodes increases to~28.
In future work, we plan to consider the integration of differential privacy into R\&A D-FL and analyze its impact on the convergence and robustness of D-FL systems under imperfect communications.

\appendix
\subsection{Proof of \textbf{Lemma \ref{theo1}}}\label{appendix:theorem proof}
Let $\mathbf{g}^{t}_{n,i}=\nabla F_{n}\left(\boldsymbol{\omega}^{t}_{n,i}\right)$ denote the gradient of client $n$ at the $i$-th epoch of the $t$-th training round of D-FL. \eqref{epoch_imperfect} can be rewritten as $\boldsymbol{\omega}^{t}_{n,i}=\boldsymbol{\omega}^{t}_{n,i-1}-\eta\mathbf{g}^{t}_{n,i-1}$.
 At the end of the $i$-th ($i=0,1,\cdots,I$) epoch,  
the virtually aggregated gradient is $\bar{\mathbf{g}}^{t}_{i}={\sum}_{ n \in \mathcal{V}}p_{n}\mathbf{g}^{t}_{n,i}$. 
The virtually aggregated global model is $\bar{\boldsymbol{\omega}}^{t}_{i}={\sum}_{ n \in \mathcal{V}}p_{n}\boldsymbol{\omega}^{t}_{n,i}$, as given by
\begin{align}
\!\!\!\!\bar{\boldsymbol{\omega}}^{t}_{i}\!\!=\!\!{\sum}_{ n \in \mathcal{V}}p_{n}\boldsymbol{\omega}^{t}_{n,i\!-\!1}\!\!-\!\!{\sum}_{ n \in \mathcal{V}}p_{n}\eta\mathbf{g}^{t}_{n,i\!-\!1}\!\!=\!\!\bar{\boldsymbol{\omega}}^{t}_{i\!-\!1}\!\!-\!\!\eta \bar{\mathbf{g}}^{t}_{i\!-\!1}.
\label{imprecise x}
\end{align}
The Euclidean distance between
$\bar{\boldsymbol{\omega}}^{t}_{i}$
and $\mathbf{w}^*$ is denoted by
\begin{align}
\Delta\boldsymbol{\omega}^{t}_{i}=\left\Vert \bar{\boldsymbol{\omega}}^{t}_{i}-\mathbf{w}^{*}\right\Vert ^{2},  \, i=0,1,2,\cdots,I.\label{eq:xre}
\end{align}

\begin{lemma}\label{appedix_lemma3}
		 The weighted Euclidean distance between $\bar{\boldsymbol{\omega}}^{t}_i$ and $\boldsymbol{\omega}_{n,i}^{t}, \forall n$ is upper bounded by
		{\small\begin{align}
		\notag &\!\!\!\!{\sum}_{ n \in \mathcal{V}}\!p_{n}\!\Vert \bar{\boldsymbol{\omega}}_{i}^{t}-\boldsymbol{\omega}_{n,i}^{t}\Vert ^{2} \leq \frac{\big((\!1\!+\!\eta)\!(\!1\!+\!4L^{2}\eta)\big)^{i\!}\!-\!1}{4L^{2}\eta+4L^{2}+1}2(\!1\!+\!\eta)\bar{\sigma}^{2}\\&\!\!\!\!+\!\!\left(\!(\!1\!\!+\!\!\eta)\!(\!1\!\!+\!\!4L^{2}\eta)\!\right)^{\!i}\!\!\Big[\sum_{ n\in\mathcal{V}}\!p_{n}\!\Vert\mathbf{w}_{n}^{t\!-\!1}\!\!-\!\!\boldsymbol{\bar{\omega}}^{t\!-\!1}\Vert^{2}\!\!\!-\!\!\big\Vert\!\!\sum_{ n\in\mathcal{V}}\!p_{n}\!\mathbf{w}_{n}^{t\!-\!1}\!\!-\!\!\boldsymbol{\bar{\omega}}^{t\!-\!1}\!\big\Vert^{2}\!\Big]\!.\label{var}%
		\end{align}%
  }%
	\end{lemma}%
	\begin{proof}%
	    See \textbf{Appendix \ref{proof:lemma5}}.
	\end{proof}
\begin{lemma}\label{appedix_lemma2}An upper bound on $\Delta\boldsymbol{\omega}^{t}_{i}$ is given as follows.%
\begin{subequations}\small\label{delta_i}
    \begin{align}
 &\!\!\!\!\!\text{When }i=1,\notag\\
&\notag\!\!\!\!\!\Delta\boldsymbol{\omega}_{1}^{t}\!\!\leq\!(1+\tau_{\varrho})\!\left(\!1\!-\!2\mu\eta+\eta^{2}L^{2}\right)\!\Delta\boldsymbol{\omega}_{I}^{t\!-\!1}\\
&\!\!\!\!+\!\!(1\!\!+\!\!\frac{1}{\tau_{\varrho}}\!)\!(1\!\!+\!\!\eta L)\!\!\Big[\!\big\Vert\!\!\sum_{n\in\mathcal{V}}\!p_{n}\!\mathbf{w}_{n}^{t\!-\!1}\!\!-\!\!\boldsymbol{\bar{\omega}}_{I}^{t\!-\!1}\!\big\Vert^{2}\!\!\!+\!\!\eta L\!\!\sum_{n\in\mathcal{V}}\!p_{n}\!\!\left\Vert \mathbf{w}_{n}^{t\!-\!1}\!\!\!-\!\!\boldsymbol{\bar{\omega}}_{I}^{t\!-\!1}\!\right\Vert^{2}\!\Big]\!;\label{delta_{1}}\\
\notag
&\!\!\!\!\!\text{When } i=2,3,\cdots,I,\\
\notag
&\!\!\!\!\!\Delta\boldsymbol{\omega}_{i}^{t}\!\!\leq\!\!\big(\!1\!\!-\!\!\frac{3\mu\eta}{2}\!\!+\!\!2L\mu\eta^2\big)\!\Delta\boldsymbol{\omega}_{i\!-\!1}^{t}\!\!+\!\!(\!2\eta^{2}L^{2}\!\!+\!\!(L\!\!+\!\!u)\eta)\!\left(\!(\!1\!\!+\!\!\eta)\!(\!1\!\!+\!\!4L^{2}\!\eta)\!\right)^{i}\\\notag&\quad\quad\times\Big(\!{\sum}_{ n\in\mathcal{V}}p_{n}\!\Vert\!\mathbf{w}_{n}^{t\!-\!1}\!\!-\!\boldsymbol{\bar{\omega}}_{I}^{t\!-\!1}\!\Vert^{2}\!-\!\big\Vert\!{\sum}_{ m\in\mathcal{V}}p_{m}\!\mathbf{w}_{n}^{t\!-\!1}\!\!-\!\boldsymbol{\bar{\omega}}_{I}^{t\!-\!1}\!\big\Vert^{2}\!\Big)\!\\
&\!\!\!\!+(\!2\eta^{2}L^{2}\!\!+\!\!(L\!\!+\!\!u)\eta)\frac{\big((\!1\!+\!\eta)\!(\!1\!+\!4L^{2}\eta)\big)^{i\!}\!-\!1}{4L^{2}\eta+4L^{2}+1}2(\!1\!+\!\eta)\bar{\sigma}^{2}.\label{delta_ib}%
\end{align}%
\end{subequations}%
	\end{lemma}%
 \begin{proof}%
See \textbf{Appendix \ref{proof:lemma6}}.
 \end{proof}	
By combining \eqref{delta_i} over $i=1,\cdots,I$, it follows that\begin{small}
	    \begin{align}
		\notag\!\Delta\boldsymbol{\omega}_{I}^{t}\!\leq&(2\eta^{2}L^{2}\!+\!(L\!+\!u)\eta)\left(\!(\!1\!+\!\eta)\!(\!1\!+\!4L^{2}\eta)\!\right)^{I}\\\notag&\times\Big(\!{\sum}_{ n\in\mathcal{V}}\!\!p_{n}\!\Vert\!\mathbf{w}_{n}^{t\!-\!1}\!\!-\!\boldsymbol{\bar{\omega}}_{I}^{t\!-\!1}\!\Vert^{2}\!-\!\big\Vert\!{\sum}_{ m\in\mathcal{V}}\!p_{m}\!\mathbf{w}_{n}^{t\!-\!1}\!\!-\!\boldsymbol{\bar{\omega}}_{I}^{t\!-\!1}\!\big\Vert^{2}\!\Big)\!\\\notag
&+(\!2\eta^{2}L^{2}\!\!+\!\!(L\!\!+\!\!u)\eta)\frac{\big((\!1\!+\!\eta)\!(\!1\!+\!4L^{2}\eta)\big)^{I\!}\!-\!1}{4L^{2}\eta+4L^{2}+1}2(\!1\!+\!\eta)\bar{\sigma}^{2}\\&+\big(\!1\!\!-\!\!\frac{3\mu\eta}{2}\!+2L\mu\eta^2\big)\Delta\boldsymbol{\omega}_{I-1}^{t}\!\notag
  \\\notag\leq&\big(\!1\!\!-\!\!\frac{3\mu\eta}{2}\!+2L\mu\eta^{2}\big)^{I-1}\Delta\boldsymbol{\omega}_{1}^{t}\!+\zeta_2\bar{\sigma}^2\\\notag
&+\!(\!2\eta^{2}L^{2}\!+\!(L\!+\!u)\eta)\left((\!1\!+\!\eta)\!(\!1\!+\!4L^{2}\eta)\right)^{2}\\\notag&\times\frac{\left((\!1\!+\!\eta)\!(\!1\!+\!4L^{2}\eta)\right)^{I-1}-\left(\!1\!\!-\!\!\frac{3\mu\eta}{2}\!+2L\mu\eta^{2}\right)^{I-1}}{(\!1\!+\!\eta)\!(\!1\!+\!4L^{2}\eta)-(1\!\!-\!\!\frac{3\mu\eta}{2}\!+2L\mu\eta^{2})}\\\notag&\times\Big(\!\sum_{ n\in\mathcal{V}}\!p_{n}\Vert\mathbf{w}_{n}^{t\!-\!1}\!\!-\!\boldsymbol{\bar{\omega}}_{I}^{t\!-\!1}\Vert^{2}\!\!-\!\big\Vert\!\!\sum_{ n\in\mathcal{V}}\!p_{n}\!\mathbf{w}_{n}^{t\!-\!1}\!\!-\!\boldsymbol{\bar{\omega}}_{I}^{t\!-\!1}\big\Vert^{2}\Big)
  \\\notag\leq&\big(\!1\!\!-\!\!\frac{3\mu\eta}{2}\!\!+\!\!2L\mu\eta^{2}\big)^{I\!-\!1}(1\!\!+\!\!\tau_{\varrho})\!\left(\!1\!\!-\!\!2\mu\eta\!\!+\!\!\eta^{2}L^{2}\right)\!\Delta\boldsymbol{\omega}_{I}^{t\!-\!1}\!\!+\!\!\zeta_2\bar{\sigma}^2\\
  &\notag\!\!\!\!\!\!+\big(\!1\!\!-\!\!\frac{3\mu\eta}{2}\!+2L\mu\eta^{2}\big)^{I-1}\!(1\!+\!\frac{1}{\tau_{\varrho}}\!)(1\!\!+\!\!\eta L)\!\\\notag&\!\!\!\!\!\!\times\!\Big[\!\big\Vert\!\!{\sum}_{n\in\mathcal{V}}\!p_{n}\!\mathbf{w}_{n}^{t\!-\!1}\!\!-\!\!\boldsymbol{\bar{\omega}}_{I}^{t\!-\!1}\!\big\Vert^{2}\!\!\!+\!\!\eta L\!\!{\sum}_{n\in\mathcal{V}}\!p_{n}\!\!\left\Vert \mathbf{w}_{n}^{t\!-\!1}\!\!\!-\!\!\boldsymbol{\bar{\omega}}_{I}^{t\!-\!1}\!\right\Vert^{2}\!\Big]\!\\
  &\!\!\!\!\!\!\!+\!\!\zeta_4\Big(\!{\sum}_{ n\in\mathcal{V}}\!p_{n}\Vert\mathbf{w}_{n}^{t\!-\!1}\!\!-\!\boldsymbol{\bar{\omega}}_{I}^{t\!-\!1}\Vert^{2}\!\!-\!\big\Vert\!\!{\sum}_{ n\in\mathcal{V}}\!p_{n}\!\mathbf{w}_{n}^{t\!-\!1}\!\!-\!\boldsymbol{\bar{\omega}}_{I}^{t\!-\!1}\big\Vert^{2}\Big)\notag
		\end{align}
	\end{small}
which can be further reorganized into
  {\small\begin{align}
		\notag&\!\!\left\Vert  \boldsymbol{\bar{\omega}}_I^{t}\!\!-\!\!\mathbf{w}^{*}\right\Vert ^{2}\!\!\leq \! \zeta_1\!\left\Vert  \boldsymbol{\bar{\omega}}_I^{t\!-\!1}\!-\!\mathbf{w}^{*}\right\Vert ^{2}\!\!+\!\zeta_2\bar{\sigma}^2\\\notag
  &+\!\zeta_3 \Big[\Big\Vert\!{\sum}_{n\in\mathcal{V}}\!p_{n}\!(\mathbf{w}_{n}^{t\!-\!1}\!\!-\!\!\boldsymbol{\bar{\omega}}_{I}^{t\!-\!1})\!\Big\Vert^{2}\!\!\!+\!\!\eta L\!\!{\sum}_{n\in\mathcal{V}}\!p_{n}\!\!\left\Vert \mathbf{w}_{n}^{t\!-\!1}\!\!\!-\!\!\boldsymbol{\bar{\omega}}_{I}^{t\!-\!1}\!\right\Vert^{2}\Big]\\
		&+\!\zeta_4\Big(\!{\sum}_{ n\in\mathcal{V}}\!p_{n}\Vert\mathbf{w}_{n}^{t\!-\!1}\!\!-\!\boldsymbol{\bar{\omega}}_{I}^{t\!-\!1}\Vert^{2}\!\!-\!\Big\Vert{\sum}_{ n\in\mathcal{V}}\!p_{n}\!\mathbf{w}_{n}^{t\!-\!1}\!\!-\!\boldsymbol{\bar{\omega}}_{I}^{t\!-\!1}\Big\Vert^{2}\Big),\label{theorem1}%
\end{align}%
}%
since $\Delta\boldsymbol{\omega}^{t}_{I}=\left\Vert  \boldsymbol{\bar{\omega}}_I^{t}\!\!-\!\!\mathbf{w}^{*}\right\Vert^2$ according to the definitions in  \eqref{accurate_model},  \eqref{imprecise x}, and \eqref{eq:xre}.
Since expectation is a linear operation,~\eqref{theorem_expectation} can be achieved by replacing both sides of \eqref{theorem1} with their expectations. Although \textbf{Lemma \ref{theo1}} is established based on full-batch gradient descent, it can be readily extended to mini-batch stochastic gradient descent, since the only difference would be the local gradient divergence between the full gradient and the stochastic gradient, which is typically upper bounded and can be characterized deterministically~\cite{8455532}.

\subsection{Proof of \textbf{Lemma \ref{lemma1}}}\label{proof_lemma1}
The expectation of the square of $\lambda_{m,n,l}^{t-1}$ in the bias matrix $\Lambda_l^{t-1}$ is given in \eqref{bar_p},
\begin{figure*}
\begin{subequations}\label{bar_p}\small
\begin{align}\mathbb{E}\left((\lambda_{m,n,l}^{t\!-\!1})^2\right)\!\!
&=\!\Pr(e_{m,n,l}^{t\!-\!1}\!=\!1)\!\!\!\!\underset{\forall\mathcal{V}'\subseteq\mathcal{V}_{m}}{\sum}\!\!\!\!\!\!\Pr(e_{k,n,l}^{t\!-\!1}\!=\!1,\forall k\!\in\!\mathcal{V}')\Pr(e_{k',n,l}^{t\!-\!1}\!\!=\!\!0,\forall k'\in\mathcal{V}_{m}\backslash\mathcal{V}')\Big(\frac{p_{m}}{\underset{ \forall k\in\mathcal{V}'}{\sum}p_{k}\!\!+\!\!p_m}\!\!-\!\!p_m\Big)^2\!\!+\!\!\Pr(e_{m,n,l}^{t\!-\!1}\!=\!0)p_m^2\label{bar_p_a}\\
\label{bar_p_d}&\!\!\!\!\!\!\!\!=\!\varrho_{m,n}^{t-1}\underset{\forall\mathcal{V}'\subseteq\mathcal{V}_{m}}{\sum}\!\!\!\Big(\underset{\forall k\in\mathcal{V}'}{\prod}\!\!\varrho_{k,n}^{t-1}\!\!\underset{ \forall k'\in\mathcal{V}_{m}\backslash\mathcal{V}'}{\prod}\!\!\!\left(1\!-\!\varrho_{k',n}^{t-1}\right)\!\!\Big)\Bigg(\frac{p_{m}\underset{ \forall k'\in\mathcal{V}_{m}\backslash\mathcal{V}'}{\sum}p_{k'}}{\underset{ \forall k\in\mathcal{V}'}{\sum}p_{k}+p_m}\Bigg)^2+(1-\varrho_{m,n}^{t-1})p_m^2\\
\label{bar_p_e}&\!\!\!\!\!\!\!\!\leq\!\varrho_{m,n}^{t-1}\underset{\forall\mathcal{V}'\subseteq\mathcal{V}_{m}}{\sum}\!\!\!\Big(\underset{\forall k\in\mathcal{V}'}{\prod}\!\!\varrho_{k,n}^{t-1}\!\!\underset{\forall k'\in\mathcal{V}_{m}\backslash\mathcal{V}'}{\prod}\!\!\!\left(1\!-\!\varrho_{k',n}^{t-1}\right)\!\!\Big)\frac{p_{m}\underset{\forall k'\in\mathcal{V}_{m}\backslash\mathcal{V}'}{\sum}p_{k'}}{\underset{\forall k\in\mathcal{V}'}{\sum}p_{k}+p_m}+(1-\varrho_{m,n}^{t-1})p_m^2\\
\label{bar_p_f}   &\!\!\!\!\!\!\!\!=\varrho_{m,n}^{t-1}p_{m}\underset{\forall j\in\mathcal{V}_{m}}{\sum}p_{j}(1-\varrho_{j,n}^{t-1})\underset{\forall\mathcal{V}'\subseteq\mathcal{V}_{m}\backslash\{j\}}{\sum}\!\!\!\Big(\underset{\forall k\in\mathcal{V}'}{\prod}\!\!\varrho_{k,n}^{t-1}\!\!\underset{\forall k'\in\mathcal{V}_{m}\backslash\{j,\mathcal{V}'\}}{\prod}\!\!\!\left(1\!-\!\varrho_{k',n}^{t-1}\right)\!\!\Big)\frac{1}{\underset{\forall k\in\mathcal{V}'}{\sum}p_{k}+p_m}+(1-\varrho_{m,n}^{t-1})p_m^2\\
\label{bar_p_g}&\!\!\!\!\!\!\!\!=\underset{\forall j\in\mathcal{V}_{m}}{\sum}p_{j}(1-\varrho_{j,n}^{t-1})\!\!\!\!\underset{\forall\mathcal{V}'\subseteq\mathcal{V}_{j};m\in\mathcal{V}'}{\sum}\bigg[\frac{p_{m}}{\underset{\forall k\in\mathcal{V}'}{\sum}p_{k}}\!\underset{\forall k\in\mathcal{V}'}{\prod}\!\varrho_{k,n}^{t-1}\underset{\forall k'\in\mathcal{V}_{j}\backslash\mathcal{V}'}{\prod}\!\!\!\left(1\!-\!\varrho_{k',n}^{t-1}\right)\!\bigg]+(1-\varrho_{m,n}^{t-1})p_m^2
\end{align}
\end{subequations}\hrulefill
\end{figure*}
where \eqref{bar_p_a} accounts for all possible scenarios under the condition that client $m$ transmits error-free or erroneous packets to client $n$; $\mathcal{V}_{m}\triangleq \mathcal{V}\backslash \{m\},\, \forall m \in \mathcal{V}$;
\eqref{bar_p_d} is obtained since
\begin{subequations}\small
    \begin{align}
       & \Pr(e_{m,n,l}^{t-1}=1)=\varrho_{m,n}^{t-1};\label{eq:bar_b_a}\\
       &\underset{ k\in\mathcal{V}}{\sum}p_{k}=\underset{ k\in\mathcal{V}'}{\sum}p_{k}+p_m+\underset{ k'\in\mathcal{V}_{m}\backslash\mathcal{V}'}{\sum}p_{k'}=1, \forall \mathcal{V}'\subseteq\mathcal{V}_m;\label{eq:bar_b_b}\\
   &\Pr(e_{k,n,l}^{t-1}\!\!=\!\!1,\!\forall k\!\in\!\!\mathcal{V}',\mathcal{V}'\! \subseteq\!\mathcal{V}_{m})\!=
   \underset{ k\in\mathcal{V}'}{\prod}\!\!\varrho_{k,n}^{t-1};\label{eq:bar_b_c1} \\
   &\Pr(\!e_{k'\!,n,l}^{t-1}\!\!=\!0,\!\forall k'\!\!\in\!\!\mathcal{V}_{m}\backslash\mathcal{V}')\!=
   \underset{ k'\in\mathcal{V}_{m}\backslash\mathcal{V}'}{\prod}\!\!\!\left(1\!-\!\varrho_{k',n}^{t-1}\right).\label{eq:bar_b_c2}
    \end{align}
\end{subequations}
Moreover, \eqref{bar_p_e} is because $0\leq\frac{p_{m}\underset{ k'\in\mathcal{V}_{m}\backslash\mathcal{V}'}{\sum}p_{k'}}{\underset{ k\in\mathcal{V}'}{\sum}p_{k}+p_m}\leq1-p_m$ and hence $\Bigg(\frac{p_{m}\underset{ k'\in\mathcal{V}_{m}\backslash\mathcal{V}'}{\sum}p_{k'}}{\underset{ k\in\mathcal{V}'}{\sum}p_{k}+p_m}\Bigg)^2\leq\frac{p_{m}\underset{ k'\in\mathcal{V}_{m}\backslash\mathcal{V}'}{\sum}p_{k'}}{\underset{ k\in\mathcal{V}'}{\sum}p_{k}+p_m}$;
and \eqref{bar_p_f} is obtained by substituting 
$\underset{\forall\mathcal{V}'\subseteq\mathcal{V}_{m}}{\sum}\underset{ k\in\mathcal{V}'}{\prod}\varrho_{k,n}^{t-1}\underset{ k'\in\mathcal{V}_{m}\backslash\mathcal{V}'}{\prod}(1\!-\!\varrho_{k',n}^{t-1})=1$ 
into the last term of \eqref{bar_p_e} and then reorganizing it. 
Finally,~\eqref{bar_p_g}, or in other words,~\eqref{eq:bound_of_lambda_nm}, is obtained.
The upper bound on $\mathbb{E}\{\left\Vert \Lambda_l^{t-1}\right\Vert ^{2}\}$ is given in \eqref{lamda_t},
 \begin{figure*}
     \begin{subequations}\label{lamda_t}\small
  \begin{align}
  \!\!\! \!\!\!\!\mathbb{E}\Big\{\left\Vert \Lambda_l^{t-1}\right\Vert ^{2}\Big\}
\leq &\underset{ n \in \mathcal{V}}{\sum}\Big(\underset{ m\in \mathcal{V}}{\sum}\mathbb{E}\big((\lambda_{m,n,l}^{t-1})^{2}\big)\Big)\big(\underset{ m\in \mathcal{V}}{\sum}(x_{m})^{2}\big)
        =\underset{ n \in \mathcal{V}}{\sum}\Big(\underset{ m\in \mathcal{V}}{\sum}\mathbb{E}(\lambda_{m,n,l}^{t-1})^{2}\Big)
       \label{lamda_t_a}\\ 
       \leq&\!\!\underset{ n \in \mathcal{V}}{\sum}\underset{ m\in \mathcal{V}}{\sum}\!\bigg((1-\varrho_{m,n}^{t-1})p_m^2+\!\!\!\!\underset{\forall j\neq m}{\sum}p_{j}(1\!-\!\varrho^{t-1}_{j,n})\!\!\!\!\!\!\!\!\underset{\forall\mathcal{V}'\subseteq\mathcal{V}_{j};m\in\mathcal{V}'}{\sum}\frac{p_{m}}{\underset{ k\in\mathcal{V}'}{\sum}p_{k}}\Big(\underset{ k\in\mathcal{V}'}{\prod}\varrho^{t-1}_{k,n}\underset{ k'\in\mathcal{V}_{j}\backslash\mathcal{V}'}{\prod}\!\!\!\!\!(1-\varrho^{t-1}_{k',n})\!\Big)\!\bigg)\label{lamda_t_b}\\
      =&\!\underset{ n \in \mathcal{V}}{\sum}\underset{ m\in \mathcal{V}}{\sum}\!(1-\varrho_{m,n}^{t-1})p_m^2\!\!+\!\!\underset{ n \in \mathcal{V}}{\sum}\bigg(\!\underset{ m\in \mathcal{V}}{\sum}\underset{\forall j\neq m}{\sum}\!p_{j}(1\!-\!\varrho^{t-1}_{j,n})\!\!\!\!\underset{\forall\mathcal{V}'\subseteq\mathcal{V}_{j};m\in\mathcal{V}'}{\sum}\!\frac{p_{m}}{\underset{ k\in\mathcal{V}'}{\sum}p_{k}}\!\underset{ k\in\mathcal{V}'}{\prod}\varrho^{t-1}_{k,n}\!\underset{ k'\in\mathcal{V}_{j}\backslash\mathcal{V}'}{\prod}\!\!\!\!(1\!-\!\varrho^{t-1}_{k',n})\!\bigg)\label{lamda_t_d}\\
    =&\!\underset{ n \in \mathcal{V}}{\sum}\underset{ m\in \mathcal{V}}{\sum}\!(1-\varrho_{m,n}^{t-1})p_m^2\!\!+\underset{ n \in \mathcal{V}}{\sum}\Big(\underset{ m \in \mathcal{V}}{\sum}p_{m}(1-\varrho^{t-1}_{m,n})\Big)=\!\underset{ n \in \mathcal{V}}{\sum}\underset{ m\in \mathcal{V}}{\sum}\!(1-\varrho_{m,n}^{t-1})\big(p_m^2+p_m\big)\label{lamda_t_g}
    \end{align}
\end{subequations}\hrulefill
 \end{figure*}
 where \eqref{lamda_t_a} is based on the Cauchy–Schwarz inequality and $\left\Vert \mathbf{x}\right\Vert =\big({\sum}_{ m\in \mathcal{V}}x_{m}^{2}\big)^{\frac{1}{2}}=1$; \eqref{lamda_t_b} is acquired by substituting \eqref{bar_p_g} into \eqref{lamda_t_a}; \eqref{lamda_t_d} expands \eqref{lamda_t_b};  and \eqref{lamda_t_g} is obtained
 by applying $\frac{{\sum}_{ m\in\mathcal{V}'}p_{m}}{{\sum}_{ k\in\mathcal{V}'}p_{k}}=1$ and $\underset{\forall\mathcal{V}'\subseteq\mathcal{V}_{j};n\in\mathcal{V}'}{\sum}\underset{ k\in\mathcal{V}'}{\prod}\varrho^{t-1}_{k,n}\underset{ k'\in\mathcal{V}_{j}\backslash\mathcal{V}'}{\prod}\left(1-\varrho^{t-1}_{k',n}\right)=\underset{ k\in\mathcal{V}_{j}}{\prod}\left(\varrho_{k,n}^{t-1}+1-\varrho_{k,n}^{t-1}\right)=1$ sequentially.
 As a result,~\eqref{lamda_t_g}, or in other words,~\eqref{eq:bound_of_lambda}, is obtained.

\subsection{Proof of \textbf{Lemma \ref{appedix_lemma3}}}
\label{proof:lemma5}
An upper bound on ${\sum}_{ n \in \mathcal{V}}p_{n}\!\left\Vert \bar{\boldsymbol{\omega}}_{i}^{t}\!-\!\boldsymbol{\omega}_{n,i}^{t}\!\right\Vert ^{2}$ is given by

\begin{subequations}\small\label{proof_lemma3_1}
     \begin{align}
&\sum_{ n \in \mathcal{V}}\!p_{n}\!\Vert \bar{\boldsymbol{\omega}}_{i}^{t}\!\!-\!\!\boldsymbol{\omega}_{n,i}^{t}\Vert ^{2}
\!\!=\!\!\sum_{ n \in \mathcal{V}}p_{n}\Vert \bar{\boldsymbol{\omega}}^{t}_{i\!-\!1}\!\!-\!\!\eta \bar{\mathbf{g}}^{t}_{i\!-\!1}\!\!-\!\!\boldsymbol{\omega}^{t}_{n,i\!-\!1}\!\!+\!\!\eta\mathbf{g}^{t}_{n,i\!-\!1}\Vert ^{2}\label{proof_lemma3_1:a}\\
&\!\!\!\!\leq\!\!(\!1\!\!+\!\!\eta)\!\!\sum_{ n \in \mathcal{V}}\!p_{n}\!\Vert \bar{\boldsymbol{\omega}}_{i\!-\!1}^{t}\!\!-\!\!\boldsymbol{\omega}_{n,i\!-\!1}^{t}\!\Vert ^{2}\!\!+\!\!(1\!\!+\!\!\frac{1}{\eta})\!\!\sum_{ n \in \mathcal{V}}\!p_{n}\!\Vert \eta\mathbf{g}_{n,i\!-\!1}^{t}\!\!-\!\!\eta \bar{\mathbf{g}}^{t}_{i\!-\!1}\Vert ^{2}\!\!,\label{proof_lemma3_1:b}
		\end{align}
 \end{subequations}
where \eqref{proof_lemma3_1:b} is based on $\Vert\mathbf{a}+\mathbf{b}\Vert^{2}\leq(1+\eta)\Vert\mathbf{a}\Vert^{2}+(1+\frac{1}{\eta})\Vert\mathbf{b}\Vert^{2},\forall \eta>0$. 

The second term of \eqref{proof_lemma3_1:b} can be rewritten as
\begin{subequations}\small\label{proof_lemma3_2}
    \begin{align}
   &\notag \!\!\!\sum_{n\in\mathcal{V}}\!\!p_{n}\!\Vert \mathbf{g}_{n,i\!-\!1}^{t}\!\!-\!\!\bar{\mathbf{g}}_{i\!-\!1}^{t}\Vert ^{2}\!\!=\!\!\!\sum_{n\in\mathcal{V}}\!\!p_{n}\!\Vert \mathbf{g}_{n,i\!-\!1}^{t}\!\!-\!\!\nabla\! F_{n}\!(\bar{\boldsymbol{\omega}}_{i\!-\!1}^{t})\!\!+\!\!\nabla \!F_{n}\!(\bar{\boldsymbol{\omega}}_{i\!-\!1}^{t})\!\!-\!\!\bar{\mathbf{g}}_{i\!-\!1}^{t}\Vert ^{2}\\ &
    \leq\!\!{\sum}_{n\in\mathcal{V}}\!p_{n}\!\left(2\Vert \mathbf{g}_{n,i\!-\!1}^{t}\!\!-\!\!\nabla\! F_{n}(\bar{\boldsymbol{\omega}}_{i\!-\!1}^{t})\Vert ^{2}\!\!+\!\!2\Vert \!\nabla \! F_{n}\!(\bar{\boldsymbol{\omega}}_{i\!-\!1}^{t})\!\!-\!\!\bar{\mathbf{g}}_{i\!-\!1}^{t}\!\Vert^{2}\!\right)\label{proof_lemma3_2:a}\\&
     \leq\!\!2L^{2}\!\!\sum_{n\in\mathcal{V}}\!p_{n}\!\Vert \!\bar{\boldsymbol{\omega}}_{i\!-\!1}^{t}\!\!-\!\!\boldsymbol{\omega}_{n,i\!-\!1}^{t}\!\Vert ^{2}\!\!+\!2\bar{\sigma}^2\!\!+\!\!2\!\!\sum_{m\in\mathcal{V}}\!\!p_{m}L^{2}\Vert \!\bar{\boldsymbol{\omega}}_{i\!-\!1}^{t}\!\!-\!\!\boldsymbol{\omega}_{m,i\!-\!1}^{t}\!\Vert ^{2}\!\!\label{proof_lemma3_2:b}\\
     &=4L^{2}\!{\sum}_{n\in\mathcal{V}}\!p_{n}\!\Vert \!\bar{\boldsymbol{\omega}}_{i\!-\!1}^{t}\!\!-\!\!\boldsymbol{\omega}_{n,i\!-\!1}^{t}\!\Vert ^{2}\!\!+\!\!2\bar{\sigma}^2,\label{proof_lemma3_2:c}
\end{align}
\end{subequations}
where \eqref{proof_lemma3_2:a} is due to  $\Vert\mathbf{a}+\mathbf{b}\Vert^{2}\leq2\Vert\mathbf{a}\Vert^{2}+2\Vert\mathbf{b}\Vert^{2}$.  
The first term of \eqref{proof_lemma3_2:b} is based on the $L$-smoothness of $F_n$ and subsequently $\left\Vert \nabla F_{n}(\bar{\boldsymbol{\omega}}_{i-1}^{t})-\!\mathbf{g}_{n,i\!-\!1}^{t}\!\right\Vert ^{2}\leq L^{2}\left\Vert \bar{\boldsymbol{\omega}}_{i\!-\!1}^{t}\!-\!\boldsymbol{\omega}_{n,i\!-\!1}^{t}\right\Vert ^{2} $, while the second term of \eqref{proof_lemma3_2:a} is upper bounded by
\begin{subequations}\small\label{proof_lemma3_3}
    \begin{align}
 \notag &\!\!\!\sum_{n\in\mathcal{V}}\!p_{n}\!\Vert \!\nabla \! F_{n}\!(\!\bar{\boldsymbol{\omega}}_{i\!-\!1}^{t}\!)\!\!-\!\!\bar{\mathbf{g}}_{i\!-\!1}^{t}\!\Vert^{2} \!\! =\!\!\!\sum_{n\in\mathcal{V}}\!\!p_{n}\!\Vert \!\nabla \!F_{n}\!(\!\bar{\boldsymbol{\omega}}_{i\!-\!1}^{t}\!)\!\!-\!\!\nabla\! F\!(\!\bar{\boldsymbol{\omega}}_{i\!-\!1}^{t}\!)\!\!+\!\!\nabla \!F(\!\bar{\boldsymbol{\omega}}_{i\!-\!1}^{t}\!)\!\!-\!\!\bar{\mathbf{g}}_{i\!-\!1}^{t}\!\Vert ^{2}\\
   & \notag= \!\!\!\sum_{n\in\mathcal{V}}\!p_{n}\!\left(\Vert \!\nabla \!F_{n}(\bar{\boldsymbol{\omega}}_{i\!-\!1}^{t})\!\!-\!\!\nabla \!F(\bar{\boldsymbol{\omega}}_{i\!-\!1}^{t})\Vert ^{2} \!\!+\!\!\Vert\! \nabla \!F(\bar{\boldsymbol{\omega}}_{i\!-\!1}^{t})\!\!-\!\!\bar{\mathbf{g}}_{i\!-\!1}^{t}\Vert ^{2}\right)\\
  &\quad\quad +\!\!2\sum_{n\in\mathcal{V}}p_{n}\left\langle\! \nabla \!F_{n}(\bar{\boldsymbol{\omega}}_{i\!-\!1}^{t})\!\!-\!\!\nabla\! F(\bar{\boldsymbol{\omega}}_{i\!-\!1}^{t}),\nabla\! F(\bar{\boldsymbol{\omega}}_{i\!-\!1}^{t})\!\!-\!\!\bar{\mathbf{g}}_{i\!-\!1}^{t}\right\rangle\label{proof_lemma3_3:a}\\
  &= \!\!\sum_{n\in\mathcal{V}}\!p_{n}\!\Big(\!\Vert \!\nabla \!F_{n}(\bar{\boldsymbol{\omega}}_{i\!-\!1}^{t})\!\!-\!\!\nabla \!F(\bar{\boldsymbol{\omega}}_{i\!-\!1}^{t})\Vert ^{2} \!\!+\!\!\big\Vert\!\!\!\sum_{m\in\mathcal{V}}\!p_{m}\!(\nabla\!F_{m}\!(\bar{\boldsymbol{\omega}}_{i\!-\!1}^{t})\!\!-\!\!\mathbf{g}_{m,i\!-\!1}^{t}\!)\big\Vert^{2}\!\Big)\label{proof_lemma3_3:b}\\
  &\leq\!\! {\sum}_{n\in\mathcal{V}}\!p_{n} \!\Big(\sigma_n^2+\sum_{m\in\mathcal{V}}p_{m}\left\Vert \nabla F_{m}(\bar{\boldsymbol{\omega}}_{i-1}^{t})-\!\mathbf{g}_{m,i\!-\!1}^{t}\!\right\Vert ^{2}\Big)\label{proof_lemma3_3:c}\\&
  \leq \bar{\sigma}^2+{\sum}_{n\in\mathcal{V}}\!p_{n} L^{2}\left\Vert \bar{\boldsymbol{\omega}}_{i\!-\!1}^{t}\!-\!\boldsymbol{\omega}_{n,i\!-\!1}^{t}\right\Vert ^{2},\label{proof_lemma3_3:d}
\end{align}
\end{subequations}
where \eqref{proof_lemma3_3:a} expands $\sum_{n\in\mathcal{V}}\!p_{n}\!\Vert \mathbf{g}_{n,i\!-\!1}^{t}\!\!-\!\!\bar{\mathbf{g}}_{i\!-\!1}^{t}\!\Vert ^{2}$; \eqref{proof_lemma3_3:b} is based on ${\sum}_{n\in\mathcal{V}}p_{n}\nabla F_{n}(\bar{\boldsymbol{\omega}}_{i-1}^{t})-\nabla F(\bar{\boldsymbol{\omega}}_{i-1}^{t})=0$ and the definitions of $\nabla F(\bar{\boldsymbol{\omega}}_{i-1}^{t})$ and $\bar{\mathbf{g}}_{i-1}^{t}$; \eqref{proof_lemma3_3:c} is obtained by applying \textbf{Assumption \ref{assumption}}-4) to the first term of \eqref{proof_lemma3_3:b} and applying Cauchy's inequality and ${\sum}_{n\in\mathcal{V}}p_{n}=1$ to the second term of \eqref{proof_lemma3_3:b};  and \eqref{proof_lemma3_3:d} exploits the $L$-smoothness of~$F_n$. 

Substituting \eqref{proof_lemma3_3:d} into \eqref{proof_lemma3_2:a} yields \eqref{proof_lemma3_2:b} and \eqref{proof_lemma3_2:c}.
By further substituting \eqref{proof_lemma3_2:c} into \eqref{proof_lemma3_1:b}, it readily follows that 
\begin{align}
\notag\!\!\!{\sum}_{ n \in \mathcal{V}}\!p_{n}\!&\Vert \bar{\boldsymbol{\omega}}_{i}^{t}\!-\!\boldsymbol{\omega}_{n,i}^{t}\Vert ^{2}\!\leq2(1\!+\!\eta)\eta\bar{\sigma}^{2}\\
&+(\!1\!+\!\eta)\!(\!1\!+\!4L^{2}\eta)\!\!{\sum}_{n\in\mathcal{V}}p_{n}\!\left\Vert \bar{\boldsymbol{\omega}}_{i\!-\!1}^{t}\!\!-\!\boldsymbol{\omega}_{n,i\!-\!1}^{t}\right\Vert ^{2}\!.\label{lemma3:e}\end{align}
By recursively updating \eqref{lemma3:e} for $i$ times, we obtain \begin{align}\notag
   \sum_{ n \in \mathcal{V}}\!p_{n}\!\Vert \bar{\boldsymbol{\omega}}_{i}^{t}-&\boldsymbol{\omega}_{n,i}^{t}\Vert ^{2} \leq \frac{\big((\!1\!+\!\eta)\!(\!1\!+\!4L^{2}\eta)\big)^{i\!}\!-\!1}{4L^{2}\eta+4L^{2}+1}2(\!1\!+\!\eta)\bar{\sigma}^{2}\\&+\left(\!(\!1\!+\!\eta)\!(\!1\!+\!4L^{2}\eta)\right)^{i}\sum_{n\in\mathcal{V}}\!p_{n}\!\left\Vert \bar{\boldsymbol{\omega}}_{0}^{t}\!-\!\boldsymbol{\omega}_{n,0}^{t}\!\right\Vert ^{2}\!\!\label{lemma3:f}.
\end{align}
By substituting $\boldsymbol{\omega}_{n,0}^{t}=\mathbf{w}_{n}^{t-1}$ into \eqref{lemma3:f}, it follows that
 \begin{align}
   \notag&
   \!\!\!\!\sum_{ n \in \mathcal{V}}\!p_{n}\!\Vert \bar{\boldsymbol{\omega}}_{i}^{t}-\boldsymbol{\omega}_{n,i}^{t}\Vert ^{2} \leq \frac{\big((\!1\!+\!\eta)\!(\!1\!+\!4L^{2}\eta)\big)^{i}\!-\!1}{4L^{2}\eta+4L^{2}+1}2(\!1\!+\!\eta)\bar{\sigma}^{2}\\&\!\!\!\!+\!\!\left(\!(\!1\!+\!\eta)\!(\!1\!\!+\!\!4L^{2}\eta)\!\right)^{i}\!{\sum}_{n\in\mathcal{V}}p_{n}\!
  \big\Vert{\sum}_{m\in\mathcal{V}}p_{m}\!\mathbf{w}_{m}^{t-1}\!\!-\!\mathbf{w}_{n}^{\!t\!-\!1\!}\big\Vert ^{2}.\label{lemma3:g}
 \end{align}
  Since $\mathbb{E}(\left\Vert X-\mathbb{E}[X]\right\Vert ^{2})=\mathbb{E}(\left\Vert X\right\Vert ^{2})-\left\Vert \mathbb{E}[X]\right\Vert ^{2}$, with $X=\mathbf{w}_{n}^{t-1}-\boldsymbol{\bar{\omega}}^{t-1}$ with probability $p_n$, we have
\begin{align}
 \!\!\notag &{\sum}_{ n\in\mathcal{V}}p_{n}\Big\Vert{\sum}_{ m\in\mathcal{V}}p_{m}\mathbf{w}_{m}^{t\!-\!1}\!\!-\!\!\mathbf{w}_{n}^{t\!-\!1\!}\Big\Vert^{2}\!\!
 \\&=\!\!\!{\sum}_{ n\in\mathcal{V}}\!p_{n}\!\Vert\mathbf{w}_{n}^{t\!-\!1}\!-\!\boldsymbol{\bar{\omega}}^{t\!-\!1}\Vert^{2}\!\!\!-\!\!\big\Vert{\sum}_{ n\in\mathcal{V}}\!p_{n}\!\mathbf{w}_{n}^{t\!-\!1}\!-\!\boldsymbol{\bar{\omega}}^{t\!-\!1}\big\Vert^{2}.\label{eq:70}
\end{align}
By substituting \eqref{eq:70} into \eqref{lemma3:g}, we finally obtain \eqref{delta_ib}.

\subsection{Proof of \textbf{Lemma \ref{appedix_lemma2}}}\label{proof:lemma6}
\subsubsection{When $i=1$} We expand $\Delta\boldsymbol{\omega}^{t}_{1}$, and then apply the inequality of arithmetic and geometric means, yielding
		\begin{align}\label{delta_r1}
\Delta\boldsymbol{\omega}^{t}_{1}
\leq&(1+\frac{1}{\tau_{\varrho}})\left\Vert \bar{\boldsymbol{\omega}}^{t}_{1}-\boldsymbol{\bar{\varphi}}^{t}_{1}\right\Vert ^{2}+
(1+\tau_{\varrho})\left\Vert \boldsymbol{\bar{\varphi}}^{t}_{1}-\mathbf{w}^{*}\right\Vert ^{2},
		\end{align}	
  where $\tau_{\varrho}$ indicates the noise level of communication errors. When the channel is error-free, $\tau_{\varrho}\rightarrow0$ under the assumption of error-free model delivery.
Define the globally aggregated model at the first epoch of the $t$-th training round as
\begin{small}
    \begin{align}\label{yr1}
\boldsymbol{\bar{\varphi}}^{t}_{1}={\sum}_{ n \in \mathcal{V}}p_n\left(\boldsymbol{\bar{\omega}}_{I}^{t-1}-\eta\nabla F_{n}(\boldsymbol{\bar{\omega}}_{I}^{t-1})\right).
\end{align}
\end{small}
By substituting \eqref{imprecise x} and \eqref{yr1}, ${\left\Vert \bar{\boldsymbol{\omega}}^{t}_{1}-\boldsymbol{\varphi}^{t}_{1}\right\Vert ^{2}}$ is upper bounded:%
 \begin{subequations}\label{a}%
  \small\begin{align}	  
&\!\!\big\Vert\bar{\boldsymbol{\omega}}^{t}_{1}\!-\!\boldsymbol{\bar{\varphi}}^{t}_{1}\!\big\Vert^{2}\!\!=\!\Big\Vert\!{\sum}_{ n \in \mathcal{V}}\!p_{n}\!\left(\mathbf{w}_{n}^{t\!-\!1\!}\!-\!\eta\mathbf{g}_{n,0\!}^{t}\!-\!\boldsymbol{\bar{\omega}}_{I}^{t\!-\!1}\!\!+\!\!\eta\!\nabla \!F_{n}\!\!\left(\boldsymbol{\bar{\omega}}_{I}^{t\!-\!1}\!\right)\!\right)\!\!\Big\Vert^{2}\!\!\!\label{a_bound:eqa}\\
&\!\!\!\leq\!\!(1\!\!+\!\!\eta L)\!\Big\Vert\!\!\sum_{ n \in \mathcal{V}}\!p_{n}\!(\mathbf{w}_{n}^{t\!-\!1}\!\!-\!\!\boldsymbol{\bar{\omega}}_{I}^{\!t\!-\!1\!})\Big\Vert^{2}\!\!\!+\!\!(1\!\!+\!\!\frac{1}{\eta L})\!\Big\Vert\!\!\sum_{ n\in\mathcal{V}}\!p_{n}\!\eta\!\left(\!\mathbf{g}_{n,0}^{t}\!\!-\!\!\nabla\! F_{n}\!(\!\boldsymbol{\bar{\omega}}_{I}^{t\!-\!1}\!)\!\right)\!\!\Big\Vert^{2}\label{a_bound:c}
\\ &\!\!\!\leq\!\!(1\!\!+\!\!\eta L)\!\Big\Vert\!\!\sum_{ n \in \mathcal{V}}\!p_{n}\!(\mathbf{w}_{n}^{t\!-\!1}\!\!-\!\!\boldsymbol{\bar{\omega}}_{I}^{t-1})\!\Big\Vert^{2}\!\!\!+\!\!(1\!\!+\!\!\eta L)\eta L\!\!\sum_{ n \in \mathcal{V}}\!p_{n}\!\Vert\mathbf{w}_{n}^{t\!-\!1}\!\!-\!\!\boldsymbol{\bar{\omega}}_{I}^{t\!-\!1}\Vert ^{2}\label{a_bound:eqb}\\
&\!\!\!=\!\!(1\!\!+\!\!\eta L)\!\Big[\big\Vert\!\!\sum_{n\in\mathcal{V}}\!p_{n}\!\mathbf{w}_{n}^{t\!-\!1}\!\!-\!\!\boldsymbol{\bar{\omega}}_{I}^{t-1}\!\big\Vert^{2}\!\!\!+\!\!\eta L\!\!\sum_{n\in\mathcal{V}}\!p_{n}\!\Vert\mathbf{w}_{n}^{t\!-\!1}\!\!-\!\!\boldsymbol{\bar{\omega}}_{I}^{t\!-\!1}\Vert\!^{2}\Big].
		\end{align}%
  \end{subequations}
Here, \eqref{a_bound:c} uses $\Vert\mathbf{a}+\mathbf{b}\Vert^{2}\leq(1+\eta L)\Vert\mathbf{a}\Vert^{2}+(1+\frac{1}{\eta L})\Vert\mathbf{b}\Vert^{2}$. \eqref{a_bound:eqb} is obtained first due to the $L$-smoothness of $F_n$ and then by applying Cauchy's inequality and ${\sum}_{ n\in\mathcal{V}}p_{n}=1$, and hence $\big\Vert\sum_{ n\in\mathcal{V}}p_{n}\eta\left(\mathbf{g}_{n,0}^{t}\!-\nabla F_{n}(\!\boldsymbol{\bar{\omega}}_{I}^{t\!-1})\!\right)\big\Vert^2\!\leq\sum_{ n\in\mathcal{V}}p_{n}\big\Vert\eta\left(\mathbf{g}_{n,0}^{t}-\nabla F_{n}(\boldsymbol{\bar{\omega}}_{I}^{t-1})\right)\big\Vert^2\leq\!\eta^2 L^2\!{\sum}_{ n\in\mathcal{V}}p_{n}\big\Vert\mathbf{w}_{n}^{t-1}-\boldsymbol{\bar{\omega}}_{I}^{t-1}\big\Vert^2$. 

By substituting \eqref{yr1} in, $\left\Vert \boldsymbol{\bar{\varphi}}^{t}_{1}-\mathbf{w}^{*}\right\Vert ^{2}$ is upper bounded by
  \begin{subequations}\label{b}\small
      \begin{align}
	\notag \left\Vert \boldsymbol{\bar{\varphi}}^{t}_{1}-\mathbf{w}^{*}\right\Vert ^{2}
 & =\left\Vert \boldsymbol{\bar{\omega}}_{I}^{t-1}-\mathbf{w}^{*}\right\Vert ^{2}+\left\Vert \eta\nabla F\left(\boldsymbol{\bar{\omega}}_{I}^{t-1}\right)\right\Vert ^{2}\notag\\
&\quad\quad\quad-2\left\langle \boldsymbol{\bar{\omega}}_{I}^{t-1}-\mathbf{w}^{*},\eta\nabla F\left(\boldsymbol{\bar{\omega}}_{I}^{t-1}\right)\right\rangle  \label{b_bound:eqa}\\
			&\leq\left(1-2\mu\eta+L^2\eta^2\right)\left\Vert \boldsymbol{\bar{\omega}}_{I}^{t-1}\!-\!\mathbf{w}^{*}\right\Vert ^{2},
		\end{align}%
  \end{subequations}
where \eqref{b_bound:eqa} is due to the $L$-smoothness of $F(\cdot)$ that $\left\Vert \nabla F\left(\boldsymbol{\bar{\omega}}_{I}^{t-1}\right)\right\Vert \leq L\left\Vert \boldsymbol{\bar{\omega}}_{I}^{t-1}\!-\!\mathbf{w}^{*}\right\Vert$ and the $\mu$-strong convexity, i.e., $\left\langle \boldsymbol{\bar{\omega}}_{I}^{t-1}-\mathbf{w}^{*},\nabla F\left(\boldsymbol{\bar{\omega}}_{I}^{t-1}\right)\right\rangle\geq \mu \left\Vert \boldsymbol{\bar{\omega}}_{I}^{t-1}\!-\!\mathbf{w}^{*}\right\Vert^2$.
By substituting \eqref{a} and \eqref{b} into \eqref{delta_r1}, we obtain~(\ref{delta_i}a).

\subsubsection{When $i=2,3,\cdots,I$} We first expand $\Delta\boldsymbol{\omega}^{t}_{i}$ by applying \eqref{imprecise x} and \eqref{eq:xre}, i.e.,
		\begin{align}
\notag\Delta&\boldsymbol{\omega}^{t}_{i}= \Big\Vert{\sum}_{ n \in \mathcal{V}}p_{n}\left(\boldsymbol{\omega}^{t}_{n,i-1}-\eta\mathbf{g}^{t}_{n,i-1}\right)-\mathbf{w}^{*}\Big\Vert^{2}\\	=&\!\left\Vert\bar{\boldsymbol{\omega}}^{t}_{i-1}\! -\! \mathbf{w}^{*}\right\Vert ^{2}\!+\!\left\Vert \eta\bar{\mathbf{g}}^{t}_{i-1}\right\Vert ^{2}\;-2\eta\left\langle \bar{\boldsymbol{\omega}}^{t}_{i-1}\!-\!\mathbf{w}^{*}\!,\bar{\mathbf{g}}^{t}_{i-1}\!\right\rangle .\label{triangle i}
		\end{align}

An upper bound of $\left\Vert \eta\bar{\mathbf{g}}^{t}_{i-1}\right\Vert ^{2}$ is obtained under the $L$-smoothness of $F_n$, as given by
  \begin{subequations}\small\label{g_bound}
	    \begin{align}
		&\!\!
  \left\Vert \! \eta\bar{\mathbf{g}}_{i\!-\!1}^{t}\right\Vert ^{2}\!\!=\!\!\eta^{2}\!\Big\Vert \!{\sum}_{n\in\mathcal{V}}\!p_{n}\!\left(\mathbf{g}_{n,i\!-\!1}^{t}\!\!-\!\!\nabla F_{n}(\bar{\boldsymbol{\omega}}_{i\!-\!1}^{t})\!\right)\!\!+\!\!\nabla F(\bar{\boldsymbol{\omega}}_{i\!-\!1}^{t})\Big\Vert ^{2}\notag\\
  &\leq\!2\eta^{2}\Big\Vert \!{\sum}_{n\in\mathcal{V}}p_{n}\!\left(\mathbf{g}_{n,i\!-\!1}^{t}\!\!-\!\!\nabla F_{n}(\bar{\boldsymbol{\omega}}_{i\!-\!1}^{t})\!\right)\!\!\Big\Vert ^{2}\!\!\!\!+\!\!2\eta^{2}\!\left\Vert \nabla F(\bar{\boldsymbol{\omega}}_{i\!-\!1}^{t})\right\Vert ^{2}\label{g_bound:b}\\
 & \leq2\eta^{2}L^{2}\!{\sum}_{n\in\mathcal{V}}p_{n}\!\left\Vert \boldsymbol{\omega}_{n,i\!-\!1}^{t}\!\!-\!\!\bar{\boldsymbol{\omega}}_{i\!-\!1}^{t}\right\Vert ^{2}\!\!\!+\!\!4L\eta^{2}\left(F(\bar{\boldsymbol{\omega}}_{i\!-\!1}^{t})\!\!-\!\!F^{*}\right)\label{g_bound:c}
  ,
		\end{align}
	\end{subequations}where  
 \eqref{g_bound:b} is based on $\Vert x_1+x_2\Vert^2\leq 2\Vert x_1\Vert^2+2\Vert x_2\Vert^2$; and \eqref{g_bound:c} is due to the $L$-smoothness of $F_n$.

		Next, we expand $-2\eta\left\langle \bar{\boldsymbol{\omega}}^{t}_{i-1}-\boldsymbol{\omega}^{*},\bar{\mathbf{g}}^{t}_{i-1}\right\rangle$ as
     \begin{small} \begin{align}
		 -\!2\eta\left\langle \bar{\boldsymbol{\omega}}_{i-\!1}^{t}\!\!-\!\boldsymbol{\omega}^{*}\!,\bar{\mathbf{g}}_{i-\!1}^{t}\right\rangle \!=&\!-\!2\eta{\sum}_{n\in\mathcal{V}}p_{n}\left\langle \bar{\boldsymbol{\omega}}_{i-\!1}^{t}\!-\!\boldsymbol{\omega}_{n,i-\!1}^{t},\mathbf{g}_{n,i-\!1}^{t}\right\rangle \!\!\notag\\&\!\!\!\!\!\!\!\!\!\!\!\!\!\!\!\!\!\!\!\!-\!2\eta{\sum}_{n\in\mathcal{V}}p_{n}\left\langle \boldsymbol{\omega}_{n,i-\!1}^{t}\!-\!\boldsymbol{\omega}^{*},\mathbf{g}_{n,i-\!1}^{t}\right\rangle .\label{A_2}
		\end{align}\end{small}%
Since $F_n$ is $L$-smooth, the first term on the RHS of \eqref{A_2} is upper bounded by%
\begin{small}\begin{align}%
		\!\!\!\!\!\!-\!\!\left\langle \!\bar{\boldsymbol{\omega}}_{i\!-\!1}^{t}\!\!-\!\!\boldsymbol{\omega}_{n,i\!-\!1}^{t},\mathbf{g}_{n,i\!-\!1}^{t}\!\right\rangle \!\!\leq\!\!F_{n}\!(\boldsymbol{\omega}_{n,i\!-\!1}^{t})\!\!-\!\!F_{n}\!(\bar{\boldsymbol{\omega}}_{i\!-\!1}^{t}\!)\!\!+\!\!\frac{L}{2}\!\Vert \bar{\boldsymbol{\omega}}_{i\!-\!1}^{t}\!\!-\!\!\boldsymbol{\omega}_{n,i\!-\!1}^{t}\!\Vert ^{2}\!\!.%
  \end{align}%
  \end{small}%
Since $F_n$ is $\mu$-strongly convex, the second term on the RHS of \eqref{A_2} is upper bounded by
       \begin{small}
            \begin{align}	\!\!\!\!\!-\! \!\left\langle\!\boldsymbol{\omega}^{t}_{n,i\!-\!1}\!\!-\!\!\boldsymbol{\omega}^{*}\!,\mathbf{g}^{t}_{n,i\!-\!1}\!\right\rangle \!\!\leq\!\!-\!F_{n}(\boldsymbol{\omega}^{t}_{n,i\!-\!1})\!\!+\!\!F_{n}(\boldsymbol{\omega}^{*})\!\!-\!\!\frac{\mu}{2}\!\Vert \boldsymbol{\omega}^{t}_{n,i\!-\!1}\!\!-\!\!\boldsymbol{\omega}^{*}\Vert ^{2}.\label{A_2-2}
		\end{align}
       \end{small}
  
  By substituting \eqref{g_bound} -- \eqref{A_2-2} into \eqref{triangle i}, it follows that%
  \begin{small}
      \begin{subequations}
    \begin{align}	\notag&\!\!\!\!\!\!\Delta\boldsymbol{\omega}^{t}_{i}
  \!\!  \leq\Delta\boldsymbol{\omega}^{t}_{i\!-\!1}\!\!+\!\!2\eta^{2}L^{2}\!{\sum}_{n\in\mathcal{V}}\!p_{n}\!\left\Vert\! \boldsymbol{\omega}_{n,i\!-\!1}^{t}\!\!-\!\!\bar{\boldsymbol{\omega}}_{i\!-\!1}^{t}\!\right\Vert ^{2}\!\!\!+\!\!4L\eta^{2}\!\!\big(F(\bar{\boldsymbol{\omega}}_{i\!-\!1}^{t})\!\!-\!\!F^{*}\big)\notag\\\notag&\!\!\!\!\!-\!2\eta{\sum}_{n\in\mathcal{V}}p_{n}\!\left(\!F_{n}\!\left(\boldsymbol{\omega}^{t}_{n,i\!-\!1}\right)\!\!-\!\!F_{n}\!(\boldsymbol{\omega}^{*})\!\!+\!\!\frac{\mu}{2}\!\left\Vert \boldsymbol{\omega}^{t}_{n,i\!-\!1}\!\!-\!\!\boldsymbol{\omega}^{*}\right\Vert ^{2}\!\right)\!\\\label{delta rea}&\!\!\!\!\!-\!\!2\!\eta{\sum}_{n\in\mathcal{V}}p_{n}\!\big(F_{n}(\bar{\boldsymbol{\omega}}_{i\!-\!1}^{t}\!)\!\!-\!\!F_{n}(\boldsymbol{\omega}_{n,i\!-\!1}^{t})\!\!-\!\!\frac{L}{2}\!\left\Vert \bar{\boldsymbol{\omega}}_{i\!-\!1}^{t}\!\!-\!\!\boldsymbol{\omega}_{n,i\!-\!1}^{t}\!\right\Vert ^{2}\!\big)
			\\
			&\!\!\!\!\!\!=2\eta^{2}L^{2}\!{\sum}_{n\in\mathcal{V}}p_{n}\!\Vert \bar{\boldsymbol{\omega}}_{i\!-\!1}^{t}\!\!-\!\!\boldsymbol{\omega}_{n,i\!-\!1}^{t}\Vert ^{2}\!\!\!+\!\!(4L\eta^{2}-2\eta)\left(F(\bar{\boldsymbol{\omega}}_{i\!-\!1}^{t})\!\!-\!\!F^{*}\right)\!\notag\\&\!\!\!+\!(L\!\!+\!\!u)\eta{\sum}_{n\in\mathcal{V}}\!p_{n}\!\!\left\Vert \bar{\boldsymbol{\omega}}_{i\!-\!1}^{t}\!\!-\!\!\boldsymbol{\omega}_{n,i\!-\!1}^{t}\right\Vert ^{2}\!\!+\!\left(\!1\!\!-\!\!\frac{\mu\eta}{2}\right)\!\Delta\boldsymbol{\omega}_{i\!-\!1}^{t}\label{delta reb}\\
 &\!\!\!\!\!\!\leq (\!2\eta^{2}L^{2}\!\!+\!\!(L\!\!+\!\!u)\eta)\!\sum_{n\in\mathcal{V}}\!p_{n}\!\!\left\Vert \bar{\boldsymbol{\omega}}_{i\!-\!1}^{t}\!\!-\!\!\boldsymbol{\omega}_{n,i\!-\!1}^{t}\right\Vert ^{2}\!\!\!+\!\!\big(\!1\!\!-\!\!\frac{3\mu\eta}{2}\!\!+\!\!2L\mu\eta^2\big)
  \Delta\boldsymbol{\omega}_{i\!-\!1}^{t},\label{delta rec}
\end{align}%
\end{subequations}%
  \end{small}%
where \eqref{delta rec} is due to the $\mu$-strong convexity of $F_n$, as well as $F(\bar{\boldsymbol{\omega}}_{i\!-\!1}^{t})\!\!-\!\!F^{*}\geq \frac{\mu}{2}\left\Vert \bar{\boldsymbol{\omega}}_{i\!-\!1}^{t}\!\!-\!\!\boldsymbol{\omega}^{*}\right\Vert ^{2}$ and $4L\eta^{2}\!-\!2\eta<0$ under the assumption that $0<\eta<\frac{1}{2L}$. 
Eventually, substituting \eqref{var}  into \eqref{delta rec} yields \eqref{delta_ib}. This proof is completed.


	\bibliographystyle{IEEEtran}
	\bibliography{ciations}

\end{document}